\documentclass[journal]{IEEEtran}
\usepackage[table]{xcolor} 
\usepackage{amssymb}
\usepackage{stfloats}
\usepackage{graphicx}
\usepackage{multirow}
\usepackage{subfigure}
\usepackage{tabularx}
\usepackage{booktabs}
\usepackage{amsmath}
\usepackage{amsthm}
\usepackage{epsfig,epsf,color,balance,cite}
\usepackage{indentfirst}
\usepackage{mathrsfs}
\usepackage{dsfont}
\usepackage{utfsym}
\UseRawInputEncoding
\usepackage{easyReview}
\usepackage[justification=centering]{caption}
\usepackage[font=small,labelfont=md,labelsep=period]{caption}

\newtheorem{lemma}{Lemma}
\newtheorem{remark}{Remark}
\usepackage{algorithm, algorithmicx, algpseudocode}
\usepackage{enumerate}
\usepackage{makecell}
\newcounter{MYtempeqncnt}
\usepackage[colorlinks,bookmarksopen,bookmarksnumbered,citecolor=blue, linkcolor=blue, urlcolor=blue]{hyperref}
\newcommand{\tabincell}[2]{\begin{tabular}{@{}#1@{}}#2\end{tabular}}
\makeatletter
\DeclareRobustCommand{\cev}[1]{%
	{\mathpalette\do@cev{#1}}%
}
\newcommand{\do@cev}[2]{%
	\vbox{\offinterlineskip
		\sbox\z@{$\m@th#1 x$}%
		\ialign{##\cr
			\hidewidth\reflectbox{$\m@th#1\vec{}\mkern4mu$}\hidewidth\cr
			\noalign{\kern-\ht\z@}
			$\m@th#1#2$\cr
		}%
	}%
}
\makeatother

\begin{document}
	\title{Channel Estimation for XL-MIMO Systems with Decentralized Baseband Processing: Integrating Local Reconstruction with Global Refinement}
	\vspace{-2em}
	\author{
		\IEEEauthorblockN{ 
			Anzheng Tang\IEEEauthorrefmark{1},~\IEEEmembership{Graduate Student Member,~IEEE}, Jun-Bo Wang\IEEEauthorrefmark{1},~\IEEEmembership{Senior Member,~IEEE},\\
			Yijin Pan\IEEEauthorrefmark{1},~\IEEEmembership{Member,~IEEE},
			Cheng Zeng\IEEEauthorrefmark{2},~\IEEEmembership{Graduate Student Member,~IEEE},
			Yijian Chen\IEEEauthorrefmark{3},\\
			Hongkang Yu\IEEEauthorrefmark{3},
			Ming Xiao\IEEEauthorrefmark{4},~\IEEEmembership{Senior Member,~IEEE},
			Rodrigo C. de Lamare\IEEEauthorrefmark{5},~\IEEEmembership{Senior Member,~IEEE},\\
			and Jiangzhou Wang\IEEEauthorrefmark{1},~\IEEEmembership{Fellow,~IEEE}
		}
		\thanks{This work was supported in part by the National Natural Science Foundation of China under Grant No. 62371123, 62350710796 and 62331024, in part by ZTE Industry-University-Institute Cooperation Funds under Grant No. IA20240723011-PO0006, and in part by Research Fund of National Mobile Communications Research Laboratory Southeast University under Grant No. 2025A03. (\emph{Corresponding authors: Jun-Bo Wang and Yijin Pan})}
		\thanks{A. Tang, J.-B. Wang, Y. Pan, and J. Wang are with the National Mobile Communications Research Laboratory, Southeast University, Nanjing 210096, China. (E-mail: \{anzhengt, jbwang, panyj, j.z.wang\}@seu.edu.cn)}
		\thanks{C. Zeng is with the School of Electronic and Optical Engineering, Nanjing University of Science and Technology, Nanjing 210094, China, and he is also with the Chinese University of Hong Kong, Shenzhen, Guangdong 518172, China. (email: {czeng}@njust.edu.cn)} 
		\thanks{Y. Chen and H. Yu are with the Wireless Product Research and Development Institute, ZTE Corporation, Shenzhen 518057, China. (E-mail:\{yu.hongkang, chen.yijian\}@zte.com.cn)}
		\thanks{Ming Xiao is with the Department of Information Science and Engineering, KTH Royal Institute of Technology, 10044 Stockholm, Sweden. (E-mail: mingx@kth.se)}
		\thanks{Rodrigo C. de Lamare is with the Centre for Telecommunications Studies, Pontifical Catholic University of Rio de Janeiro, Rio de Janeiro 22451-900, Brazil, and also with the School of Physics, Engineering and Technology, University of York, York YO10 5DD, U.K. (E-mail: delamare@puc-rio.br).}
	}
	\maketitle
	\begin{abstract}
		In this paper, we investigate the channel estimation problem for extremely large-scale multiple-input
		multiple-output (XL-MIMO) systems with a hybrid analog-digital architecture, implemented within a decentralized baseband processing (DBP) framework with a star topology. Existing centralized and fully decentralized channel estimation methods face limitations due to excessive computational complexity or degraded performance. To overcome these challenges, we propose a novel two-stage channel estimation scheme that integrates local sparse reconstruction with global fusion and refinement.
		Specifically, in the first stage, by exploiting the sparsity of channels in the angular-delay domain, the local reconstruction task is formulated as a sparse signal recovery problem. To solve it, we develop a graph neural networks-enhanced sparse Bayesian learning (SBL-GNNs) algorithm, which effectively captures dependencies among channel coefficients, significantly improving estimation accuracy.
		In the second stage, the local estimates from the local processing units (LPUs) are aligned into a global angular domain for fusion at the central processing unit (CPU). 
		Based on the aggregated observations, the channel refinement is modeled as a Bayesian denoising problem. 
		To efficiently solve it, we devise a variational message passing algorithm that incorporates a Markov chain-based hierarchical sparse prior, effectively leveraging both the sparsity and the correlations of the channels in the global angular-delay domain.
		Simulation results show the effectiveness and superiority of the proposed SBL-GNNs algorithm over existing methods, demonstrating improved estimation performance and reduced computational complexity.
	\end{abstract}
	\IEEEpeerreviewmaketitle
	\section{Introduction}
	\label{section1}
	To meet the increasing demand for data transmission, millimeter-wave (mmWave) bands have become essential to alleviate spectrum congestion in traditional microwave frequencies. However, high-frequency signals suffer from significant attenuation and are highly susceptible to blockages. To overcome these challenges, massive multiple-input
	multiple-output (MIMO) systems with hybrid beamforming architectures have been widely adopted to mitigate considerable path loss.
	Recently, the concept of extremely large-scale MIMO (XL-MIMO) has emerged as an evolution of the massive MIMO paradigm, attracting substantial attention from both academia and industry \cite{Tutorial}. XL-MIMO systems are anticipated to achieve higher spectrum and energy efficiencies, expand network coverage, and enhance positioning accuracy by leveraging substantial spatial multiplexing and beamforming gains \cite{Anzheng1, An_1, lecheng}. To fully realize these benefits, advanced signal processing techniques such as beamforming and equalization must be meticulously designed. Notably, obtaining accurate channel state information (CSI) is essential for realizing these performance gains, underscoring the critical role of precise channel knowledge in maximizing the potential of XL-MIMO technologies \cite{An_2, Yang_1}.
	\vspace{-1em}
	\subsection{Existing Research}
	Recent research has extensively explored channel estimation algorithms tailored to various characteristics of XL-MIMO channels. For example, \cite{Centralized_CE_3} addressed spherical wavefront effects by designing an angular-and-distance-dependent near-field codebook and proposed a polar-domain sparsity-based simultaneous orthogonal matching pursuit (SOMP) algorithm. This approach was further extended to hybrid near- and far-field scenarios by \cite{PolarCS2, PolarCS3}. To further mitigate the size of the polar-domain codebook, \cite{Centralized_CE_4} introduced a distance-parameterized angular-domain representation, facilitating channel estimation through joint dictionary learning and sparse recovery methods.
	
	Considering the spatially non-stationary (SnS) property of XL-MIMO channels, \cite{Centralized_CE_1} transformed SnS channels into a series of spatially stationary channels using a group time block code (GTBC)-based signal extraction scheme. Leveraging the antenna-domain sparsity of SnS channels, \cite{SnS_1} characterized SnS channels with a subarray-wise Bernoulli-Gaussian prior and proposed a Bayesian inference-based scheme for joint user activity and channel estimation. Furthermore, to fully exploit the sparsity of SnS channels in the antenna-delay domain, \cite{SnS_2} introduced a structured prior model based on a hidden Markov model (HMM) and proposed the turbo orthogonal approximate message passing (Turbo-OAMP) algorithm for low-complexity channel estimation.
	Furthermore, taking into account both antenna-domain-and-angular domain characteristics, \cite{Anzheng2} introduced a two-stage approach for detecting visibility regions (VRs) and estimating channels. However, the two stages are performed separately, leading to insufficient utilization of angular-domain information during the VRs detection phase.
	To overcome the limitation, we further introduced a cascaded generalized approximate message passing (TL-GAMP) algorithm to jointly estimate VRs and channels in our subsequent work \cite{Centralized_CE_6}. 
	
	Taking into account the near-field dual-wideband effect for wideband XL-MIMO systems, \cite{Dual_wideband_1} proposed to model the angular-delay domain channels with an independent and non-identically
	distributed Bernoulli-Gaussian prior, and the channel estimation problem is solved by the constrained Bethe free energy minimization framework. To further exploit the structure sparsity of angular-delay domain channels, \cite{tang2024spatial} introduced a 4-connected Markov random filed (MRF) based sparse prior model and proposed a unity approximate message passing to achieve the low-complexity estimation.
	
	However, the aforementioned algorithms are implemented centrally. 
	{As the number of antennas increases, such centralized architecture encounters bottlenecks in terms of high interconnection cost and computational complexity. 
	To address these challenges, decentralized baseband processing (DBP) has emerged as a promising technique \cite{De_MIMO_15, De_MIMO_8, De_MIMO_5}}. 
	In the DBP architecture, the antennas are divided into several antenna clusters, and each cluster is processed by an independent and more affordable local processing unit (LPU). 
	{Notably, there are significant distinctions between the DBP architecture and cell-free systems. Specifically, in contrast to cell-free systems, where the access points (APs) are geographically distributed and potentially lack a centralized base station \cite{Yang_2, De_MIMO_1, cell_free1}, the DBP architecture retains a centralized antenna array.
	Furthermore, in the cell-free context, the APs have independent power constraints because they are geographically separated devices, each provisioned with its own power line. 
	In contrast, the DBP network is a baseband signal processing layer architecture mounted on XL-MIMO antennas, where the segments share a common power supply\cite{De_MIMO_4}.}
	
	Under the DBP architecture, it is straightforward for each LPU to perform channel estimation based solely on its locally received signals. 
	However, this fully decentralized strategy would experience significant performance degradation due to the neglect of inter-LPU correlations. Consequently, advanced signal processing techniques are required to achieve a promising estimation performance while maintaining low computational complexity.
	Recently, various signal processing algorithms for XL-MIMO systems with DBP architecture have been extensively investigated \cite{De_MIMO_7, De_MIMO_8, De_MIMO_5, De_MIMO_9, De_MIMO_10, De_MIMO_11, De_MIMO_6, De_MIMO_2}. However, relatively few works have focused on decentralized channel estimation problems \cite{De_MIMO_2, De_MIMO_6}. Specifically, \cite{De_MIMO_2} explored decentralized channel estimation with a daisy-chain network topology and proposed a distributed minimum mean square error (MMSE) algorithm leveraging spatial correlation among antenna elements. To further reduce fronthaul communication and computational costs, \cite{De_MIMO_6} proposed two low-complexity diagonal MMSE channel estimators, namely the aggregate-then-estimate (AGE) and the estimate-then-aggregate (EAG) algorithms. Unluckily, there methods are derived under the full digital precoding architecture, which is inapplicable for the more practical hybrid precoding architecture. 
	\vspace{-1em}
	\subsection{Motivations and Contributions}
	Addressing the gap in the existing literature, this paper investigates the channel estimation problem for XL-MIMO systems with a hybrid analog-digital architecture implemented under a DBP framework with a star topology. Specifically, the key contributions of this work are summarized as follows:
	\begin{itemize}
		\item To overcome the limitations of existing centralized and fully decentralized channel estimation schemes, we propose a two-stage estimation scheme comprising local sparse reconstruction and global fusion with refinement.
		In the first stage, the channel estimation for each subarray is derived as a sparse signal recovery problem by leveraging angular-delay sparsity within a Bayesian inference framework. Then, the local estimates are sent to the CPU via the fronthaul link.
		In the second stage, the local estimates are first aligned into a unified global angular domain with enhanced resolution to exploit the correlations among subarray channels. Then, the channel refinement task is formulated as a Bayesian denoising problem based on the aggregated observations.
		\item To efficiently solve the sparse signal recovery problem in LPUs, we propose a graph neural networks-based sparse Bayesian learning (SBL-GNNs) algorithm. Specifically, the expectation step (E-step) of the traditional SBL algorithm is preserved, while the non-linear functions for the prior variance update in the maximization step (M-step) are replaced by GNNs. {Based on the observations from the E-step, the prior variance update is formulated as a maximum a posteriori (MAP) estimation problem, modeled using a pairwise MRF representation. Leveraging the properties of this pairwise MRF, a GNN architecture is designed to efficiently capture dependencies among channel coefficients and perform MAP inference}, enhancing both the accuracy and robustness of channel estimation.
		\item To address the Bayesian denoising problem of aggregated channels, we propose a Markov chain-based hierarchical prior model that effectively captures both the sparsity and correlations among channel coefficients in the angular-delay domain. In this prior model, the variance state of the channel coefficients is governed by a Markov chain, which encourages structured block patterns and suppresses isolated coefficients that deviate from their neighboring values. Additionally, to efficiently perform Bayesian inference, {we devise a variational message passing algorithm tailored to the prior model, capable of handling the dependencies and loops in the factor graph}.
		
		\item Extensive numerical results are presented to validate the efficacy of the proposed algorithms in terms of computational complexity and estimation performance. Specifically, for local reconstruction, the proposed SBL-GNNs algorithm demonstrates the lowest runtime while achieving the best reconstruction performance. Moreover, the overall computational complexities of the proposed two-stage scheme are significantly lower than those of the centralized scheme, while maintaining performance comparable to that of the centralized approach.
	\end{itemize}
	
	Organization: The rest of the paper is organized as follows. In Section \ref{section2}, we introduce the system model and problem description. Section \ref{section3} describes the proposed two-stage estimation scheme. In Section \ref{section4}, we detail the proposed SBL-GNNs algorithm. Section \ref{section5} introduces the proposed centralized refinement algorithm for the CPU. Finally, simulation results are presented in Section \ref{section6}, and conclusions are drawn in Section \ref{section7}.
	
	Notations: Lower-case letters, bold-face lower-case letters, and bold-face upper-case letters are used for scalars, vectors and matrices, respectively; 
	The superscripts $\left(\cdot \right)^{\mathrm{T}}$ and $\left(\cdot \right)^{\mathrm{H}}$ stand for transpose and conjugate transpose, respectively;  
	$\mathrm{blkdiag}\left(\mathbf{X}_1, \mathbf{X}_2, \dotso, \mathbf{X}_N\right)$ denotes a diagonal matrix with $\left\{\mathbf{X}_1, \mathbf{X}_2, \dotso, \mathbf{X}_N\right\}$ being its diagonal elements; 
	{$\otimes$ denotes the Kronecker product; $\mathrm{vec}(\cdot)$ denotes the vectorization operations, which reorganizes the matrix into a column vector.}
	$\mathbb{C}^{M \times N}$ denotes a $M \times N$ complex matrix. 
	In addition, a random variable $x \in \mathbb{C}$ drawn from the complex Gaussian distribution with mean $m$ and variance $v$ is characterized by the probability density function $\mathcal{CN}(x;m,v) = {\exp\left\lbrace-\{|x-m\right|^2}/{v}\}/{\pi v} $; a random variable $\gamma \in \mathbb{R}$ from Gamma distribution with mean $a/b$ and variance $a/b^2$ is characterized by $\mathcal{G}a(\gamma; a, b) \propto \gamma^{a-1}\exp(-\gamma b)$.
	\section{System Models and Problem Description}
	\label{section2}
	\begin{figure}
		\centering
		\includegraphics[width=0.45\textwidth]{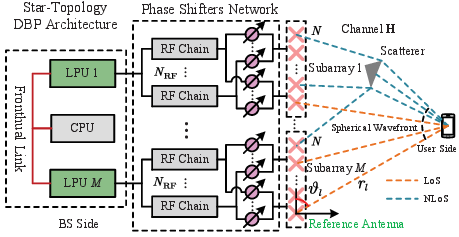}
		\caption{XL-MIMO systems with star-topology DBP architecture.}
		\label{system_model}
	\end{figure}
	{Consider a mmWave XL-MIMO system in which the base station (BS) equipped with an $N_t$-antenna extremely large-scale array (ELAA), arranged in a uniform linear configuration, serves $U$ single-antenna users. Additionally, to enable efficient implementation, the BS adopts a hybrid analog-digital architecture based on a DBP framework with a star topology \cite{De_MIMO_2,DBP_Compressive}, as shown in Fig. \ref{system_model}.}   
	In this configuration, the ELAA is divided into $M$ non-overlapping subarrays, each consisting of $N = N_{{t}} / M$ antennas. 
	{Without loss of generality, we assume the antenna spacing is $d = \lambda_c/2$, where $\lambda_c = c/f_c$ with $c$ and $f_c$ indicating the speed of light and central carrier frequency, respectively.}
	These subarrays are connected to LPUs through a phase shifter network, with each subarray associated with $N_{\mathrm{RF}}$ radio frequency (RF) chains.
	Then, the LPUs coordinate with a CPU through fronthaul links for baseband signal processing\footnote{Due to the wired transmission of the fronthaul link, we assume that the different LPUs can achieve perfect synchronization.}.
	 
	Furthermore, to mitigate the effects of frequency-selective fading, orthogonal frequency-division multiplexing (OFDM) is employed. Without loss of generality, assume that $K$ subcarriers are uniformly selected from the available subcarriers to transmit pilot symbols. {For notation conciseness later, we denote the following sets $\mathcal{M}=\{1, 2, \cdots, M\}$, $\mathcal{N}=\{1, 2, \cdots, N\}$, $\mathcal{J}=\{1, 2, \cdots, N_t\}$, $\mathcal{K}=\{1, 2, \cdots, K\}$, and {$\mathcal{V} = \{1,2,\cdots, NK\}$}.}
	\vspace{-1em}
	\subsection{Channel Model}
	In mmWave XL-MIMO systems, the deployment of ELAAs and the utilization of high-frequency bands enable near-field communications over distances of up to hundreds of meters. 
	This renders the far-field plane wavefront assumption invalid and necessitates accounting for spherical wavefront effects. 
	Additionally, dual-wideband effects are inherent in wideband XL-MIMO systems, further complicating channel modeling.
	To accurately capture these characteristics, we adopt a uniform spherical-wavefront-based channel model \cite{tang2024spatial}. Specifically, considering $G$ propagation paths between the BS and the user $u$, thus, the overall spatial-frequency channel of user $u$, denoted as $\mathbf{H}^{(u)} \in \mathbb{C}^{N_t \times K}$, is expressed as
	\begin{equation}
		\mathbf{H}^{(u)} = \sum_{g=1}^{G}\alpha_g^{(u)}
		\left(\mathbf{b}(\psi_g^{(u)}, \varphi_g^{(u)}) \mathbf{a}^{\mathrm{T}}(\tau_g^{(u)})\right)\odot\boldsymbol{\Theta}(\psi_g^{(u)}, \varphi_g^{(u)}),
		\label{H_matrix}
	\end{equation}
	where $\alpha_g^{(u)} \in \mathbb{C}$ is the equivalent complex gain of the $g$-th path for user $u$; $\mathbf{a}(\tau_g^{(u)}) \in \mathbb{C}^{K\times 1}$ and $\mathbf{b}(\psi_g^{(u)}, \varphi_g^{(u)}) \in \mathbb{C}^{N_t\times 1}$ denote frequency-domain and spatial-domain steering vectors corresponding to the $g$-th path, respectively, which are respectively given by
 	\begin{align}
 		[\mathbf{a}(\tau_g^{(u)})]_k &= \mathrm{e}^{-\mathrm{j}2\pi {f_{k}}\tau_{g}^{(u)}}, k \in  \mathcal{K}\\
 		\label{b_array}[\mathbf{b}(\psi_g^{(u)},  \varphi_g^{(u)})]_j &= \mathrm{e}^{-\mathrm{j}2\pi \left({\psi_g^{(u)}}(j-1) - {\varphi_g^{(u)}}(j-1)^2 \right)}, j\in \mathcal{J}
 	\end{align}
 	where $f_k = f_c+f_s(k-1)$ indicates the $k$-th subcarrier frequency with $f_s$ being the system bandwidth;
 	$\tau_g^{(u)} = r_g^{(u)}/c$ denotes the path delay with $r_g^{(u)}$ indicating the distance between the reference antenna and user or scatters. 
 	In addition, $\psi_g^{(u)} \triangleq d\cos \vartheta_g^{(u)} / \lambda_c$ and $\varphi_g^{(u)} \triangleq {d^2\sin^2\vartheta_g^{(u)}}/{2r_g^{(u)}\lambda_c}$ with $\vartheta_g^{(u)}$ indicating the direction angle of the $g$-th path to the reference antenna.
 	{From (\ref{b_array}), the array steering vectors $\mathbf{b}(\psi_g, \varphi_g)$ are distance-dependent and incorporate quadratic terms with respect to the antenna index $n$, accounting for variations in path delay and direction angle across antenna elements, thereby reflecting the impact of near-field propagation.}
 	$\boldsymbol{\Theta}(\psi_g^{(u)}, \varphi_g^{(u)})$ denotes the frequency-dependent phase shifts matrix, {whose $(j,k)$-th element, $j\in\mathcal{J}$ and $k\in\mathcal{K}$},  is given by 
 	\begin{equation}
 		\left[\boldsymbol{\Theta}(\psi_g^{(u)}, \varphi_g^{(u)})\right]_{j,k} = \mathrm{e}^{-\mathrm{j}2\pi {f_k}\left((j-1)\frac{\psi_g^{(u)}}{f_c}- (j-1)^2\frac{\varphi_g^{(u)}}{f_c}\right)}.
 	\end{equation}
 	{Consequently, the channel of user $u$ associated with subarray $m$ is represented by $\mathbf{H}_m^{(u)} = [\mathbf{h}_{m,1}^{(u)},\mathbf{h}_{m,2}^{(u)},\cdots, \mathbf{h}_{m,K}^{(u)}]\in \mathbb{C}^{N\times K}$, which consists of the rows of $\mathbf{H}^{(u)}$ from the $((m-1)N+1)$-th to the $mN$-th row with $m\in \mathcal{M}$.}
 	\vspace{-1em}
	\subsection{Problem Description}
	This paper mainly addresses the uplink channel estimation task. {Without loss of generality, we assume that the mutually orthogonal pilot sequences are used by different users. As a result, the channel estimation for each user is independent of the others. In this context, we focus on the channel estimation problem for an arbitrary user $u$. For notation simplification, we omit the superscript $u$ in the following.
	At the BS side, each subarray combines the received pilot signals using $N_{\mathrm{RF}}$ RF chains, each associated with a distinct beam.}
	Assume the total number of training beams as $Q$. As a result, we needs $P=Q/N_{\mathrm{RF}}$ time slots to traverse all receive beams given a fixed pilot symbol.
	Denote $\iota_{k,p}$ as the pilot symbol of the $k$-th subcarrier in the $p$-th time slot. 
	Then, the received signal $\mathbf{y}_{m,k,p} \in \mathbb{C}^{N_{\mathrm{RF}}\times 1}$ of subarray $m$ at the $k$-th subcarrier in the $p$-th time slot is given by  
	\begin{equation}
		\mathbf{y}_{m, k,p} = \mathbf{W}_{m, p}\mathbf{h}_{m,k} \iota_{k,p} + \mathbf{W}_{m, p}\mathbf{n}_{m, k, p},
	\end{equation}
	where $\mathbf{W}_{m,p} \in \mathbb{C}^{N_{\mathrm{RF}}\times N}$ denotes the combining matrix for subarray $m$ in the $p$-th time slot, $m \in \mathcal{M}$ and $p = 1,2,\cdots,P$, with each element randomly selected from $\{-1/\sqrt{N}, 1/\sqrt{N}\}$;
	{$\mathbf{h}_{m,k} \in \mathbb{C}^{N\times 1}$ and $\mathbf{n}_{m, k, p} \in \mathbb{C}^{N\times 1}$ indicate the channel and noise vectors, respectively, for subarray $m$ at the $k$-th subcarrier with $m\in\mathcal{M}$ and $k\in\mathcal{K}$. Additionally, $\mathbf{n}_{m, k, p}$ is the additional white Gaussian (AWGN) noise vector, where each element obeys circularly symmetric complex Gaussian distribution with variance $\zeta_m^{-1}$ with $\zeta_m$ indicating the variance precision parameters.}
	Without loss the generality, we assume $\iota_{k,p} = 1$ for all $k$ and $p$. 
	By collecting the received pilot signals corresponding to all $P$ time slots, the total received signal $\mathbf{y}_{m,k} = [\mathbf{y}_{m,k,1}^{\mathrm{T}}, \cdots, \mathbf{y}_{m,k,P}^{\mathrm{T}}]^{\mathrm{T}} \in \mathbb{C}^{Q\times 1}$ for subarray $m$ at the $k$-th subcarrier can be expressed as
	\begin{equation}
		\mathbf{y}_{m,k} = \mathbf{W}_{m} \mathbf{h}_{m,k} + {\mathbf{n}}_{m,k},
	\end{equation}
	where $\mathbf{W}_m= \left[\mathbf{W}_{m,1}^{\mathrm{T}}, \mathbf{W}_{m,2}^{\mathrm{T}}, \cdots, \mathbf{W}_{m,P}^{\mathrm{T}}\right]^{\mathrm{T}} \in \mathbb{C}^{Q\times N}$ is the collective combining matrix. 
	{${\mathbf{n}}_{m,k} = [\mathbf{n}_{m, k,1}^{\mathrm{T}}\mathbf{W}_{m,1}^{\mathrm{T}}, \mathbf{n}_{m, k,2}^{\mathrm{T}}\mathbf{W}_{m,2}^{\mathrm{T}}, \cdots, \mathbf{n}_{m, k,P}^{\mathrm{T}}\mathbf{W}_{m,P}^{\mathrm{T}}]^{\mathrm{T}} \in \mathbb{C}^{Q\times 1}$ is the effective noise vector with the covariance matrix $\mathrm{blkdiag}(\zeta_m^{-1}\mathbf{W}_{m,1}\mathbf{W}_{m,1}^{\mathrm{H}},\cdots,\zeta_m^{-1}\mathbf{W}_{m,P}\mathbf{W}_{m,P}^{\mathrm{H}})$. 
	It follows that the incorporation of the analog combiner matrix $\mathbf{W}_{m,p}$ can lead to a non-identity covariance matrix $\mathbf{R}$, indicating that $\mathbf{n}_{m,k}$ may exhibit colored characteristics.
	Notably, if $\mathbf{W}_{m,p}$ is a semi-unitary matrix for any $m$ and $p$, i.e., $\mathbf{W}_{m,p} \mathbf{W}_{m,p}^{\mathrm{H}} = \mathbf{I}_{N_{\mathrm{RF}}}$, the covariance matrix $\mathbf{R}$ simplifies to $\zeta_m^{-1} \mathbf{I}_{Q}$. 
	In this case, $\mathbf{n}_{m,k}$ remains strictly white noise and follows $\mathbf{n}_{m,k} \sim \mathcal{CN}(\mathbf{0}, \zeta_m^{-1} \mathbf{I}_{Q})$. Consequently, for simplicity, this manuscript assumes $\mathbf{W}_{m,p}$ to be semi-unitary for all $m$ and $p$. In practical implementations, $\mathbf{W}_{m,p}$ can be chosen as a partial DFT or Hadamard matrix.\footnote{Notably, this assumption does not affect the subsequent algorithm design. The proposed estimation scheme is also applicable for a non-orthogonal analog combiner matrix, requiring only a pre-whitening to (\ref{De_obs}) \cite{Centralized_CE_6}.}
	Further, by concentrating the received signals from
	different subcarriers, the collective pilot signal $\mathbf{Y}_m = \left[\mathbf{y}_{m,1}, \mathbf{y}_{m,2}, \cdots,
	\mathbf{y}_{m,K}\right]\in \mathbb{C}^{Q\times K}$ at the $k$-th subcarrier is given by
	\begin{equation}
		\mathbf{Y}_m = \mathbf{W}_m\mathbf{H}_m + {\mathbf{N}}_m,
		\label{Y_pilot_m}
	\end{equation}
	{where $\mathbf{H}_m=\left[\mathbf{h}_{m,1},\mathbf{h}_{m,2}, \cdots, \mathbf{h}_{m,K}\right] \in \mathbb{C}^{N\times K}$ and $\mathbf{N}_m = \left[{\mathbf{n}}_{m,1},{\mathbf{n}}_{m,2},\cdots,{\mathbf{n}}_{m,K}\right] \in \mathbb{C}^{Q\times K}$ represent the collective channel matrix and noise matrix for subarray $m$, respectively. In addition, each column of ${\mathbf{N}}_m$ is a white noise vector with identical covariance matrix $\zeta_m^{-1}\mathbf{I}_{Q}$}. 
	
	Under the DBP architecture, each LPU can directly send its received signal $\mathbf{Y}_m$ to a CPU for the centralized estimation. The received signal at the CPU is given by
	\begin{equation}
		\mathbf{Y} = \mathbf{W}\mathbf{H} + {\mathbf{N}},
		\label{Y_pilot}
	\end{equation}
	where $\mathbf{W} = \mathrm{blkdiag}\left(\mathbf{W}_1, \mathbf{W}_2, \cdots, \mathbf{W}_{M}\right) \in \mathbb{C}^{{QM\times N_{{t}}}}$
	and $\mathbf{N} = \left[\mathbf{N}_{1}^{\mathrm{T}},\mathbf{N}_{2}^{\mathrm{T}}, \cdots, \mathbf{N}_{M}^{\mathrm{T}} \right]^{\mathrm{T}} \in \mathbb{C}^{QM\times K}$ denote the overall received matrix and noise matrix, respectively. Utilizing the sparsity of $\mathbf{H}$ in the angular-delay domain \cite{tang2024spatial}, the channel estimation task can be formulated as sparse signal recovery problem, i.e., 
	\begin{equation}
		\arg \min_{\mathbf{H}_\mathrm{T}} = \frac{1}{2}\left\|\mathbf{Y} - \mathbf{W} \mathbf{F}_{\mathrm{A}}^{\mathrm{H}}\mathbf{H}_{\mathrm{T}}\mathbf{F}_{\mathrm{D}}\right\|^2_2 + \varrho \left\|\mathbf{H}_\mathrm{T} \right\|_0, \,
		\label{L0_norm}
	\end{equation}
	where $\varrho > 0$ denotes a regularization factor, $\mathbf{H}_{\mathrm{T}}$ represents the angular-delay channel, {and $\mathbf{F}_{\mathrm{A}} \in \mathbb{C}^{N_{{t}} \times N_{{t}}}$ and $\mathbf{F}_{\mathrm{D}} \in \mathbb{C}^{K \times K}$ are the discrete Fourier transform (DFT) matrices of dimensions $N_{{t}}$ and $K$, respectively.
	Up to now, several algorithms have been developed to address the problem in (\ref{L0_norm}), including the SOMP method \cite{OMP_1} and Bayesian inference-based approaches \cite{PC_SBL,Variance_State_1,Variance_State_2}}. However, with the increasing number of antennas and subcarriers in mmWave XL-MIMO systems, these centralized algorithms become computationally prohibitive \cite{De_MIMO_4}. {For example, the SBL algorithm incurs a computational complexity of $\mathcal{O}(M^3Q^3K^3)$ due to the need for matrix inversions, making it unsuitable for large-scale scenarios.}
	
	To address these challenges, it is essential to design more practical estimation schemes with reduced computational complexity. A natural approach under the DBP architecture is for each subarray to independently estimate its local channel using its received signal. However, such a fully decentralized scheme ignores the inherent channel correlations among subarrays, potentially leading to estimation performance degradation. Consequently, there is an urgent need to develop efficient schemes that strike an attractive balance between computational efficiency and estimation accuracy.
	\section{Proposed Two-Stage Estimation Scheme}
	\label{section3}
	In this section, we propose a two-stage reconstruction-then-refinement estimation scheme, as illustrated in Fig.~\ref{Scheme}.
	\subsection{Stage I: Local Sparse Channel Reconstruction} 
	In the first stage, the goal is to reconstruct the channels for each subarray by exploiting the sparsity of $\mathbf{H}_m$ in the local angular-delay domain. 
	Denote $\mathbf{F}_{\mathrm{L}} \in \mathbb{C}^{N \times N}$ as an $N$-dimension DFT matrix, $\mathbf{H}_m$ can be approximated as
	\begin{equation}
		\mathbf{H}_m = \mathbf{F}_{\mathrm{L}}^{\mathrm{H}}\mathbf{H}_{\mathrm{T},m}^{\mathrm{L}}\mathbf{F}_{\mathrm{D}}, m \in \mathcal{M},
		\label{H_ma}
	\end{equation}
	where $\mathbf{H}_{\mathrm{T},m}^{\mathrm{L}} \in \mathbb{C}^{N\times K}$ is the local angular-delay channel of subarray $m$. Due to the limited propagation paths at mmWave frequencies, only a small number of angular and delay components in $\mathbf{H}_{\mathrm{T},m}^{\mathrm{L}}$ are significant, with these components concentrated within specific angular and delay ranges \cite{tang2024spatial}. As a result, $\mathbf{H}_{\mathrm{T},m}^{\mathrm{L}}$ exhibits block sparsity. 
	Utilizing (\ref{H_ma}), (\ref{Y_pilot_m}) can be rewritten as
	\begin{equation}
		\mathbf{y}_m = \boldsymbol{\Phi}_m\mathbf{x}_m + \mathbf{n}_m, m \in \mathcal{M},
		\label{De_obs}
	\end{equation} 
	where $\mathbf{y}_m = \mathrm{vec}({\mathbf{Y}_m}) \in \mathbb{C}^{QK\times1}$ and $\mathbf{x}_m = \mathrm{vec}\left(\mathbf{H}_{\mathrm{T},m}^{\mathrm{L}}\right) \in \mathbb{C}^{NK \times 1}$ are the stacked received signal and local angular-delay channel, respectively. 
	${\boldsymbol{\Phi}}_m = \mathbf{F}_{\mathrm{D}}^{\mathrm{T}}\otimes \left(\mathbf{W}_m\mathbf{F}_{\mathrm{L}}^{\mathrm{H}} \right) \in \mathbb{C}^{QK \times NK}$ denotes the equivalent measurement matrix. 
	{Additionally, $\mathbf{n}_m = \mathrm{vec}({\mathbf{N}_m}) \sim \mathcal{CN}(\mathbf{n}_m; 0, \zeta_m^{-1} \mathbf{I}_{QK})$ denotes the noise vector with covariance $\zeta_m^{-1} \mathbf{I}_{QK}$, where $\zeta_m$ is the unknown noise precision parameter to be estimated as well. Without loss of generality, a conjugated prior probability distribution is assigned for ${\zeta}_m$, i.e., $p(\zeta_m)\propto \zeta_m^{-1}$.}
	
	\begin{figure}
		\centering
		\includegraphics[width=0.4\textwidth]{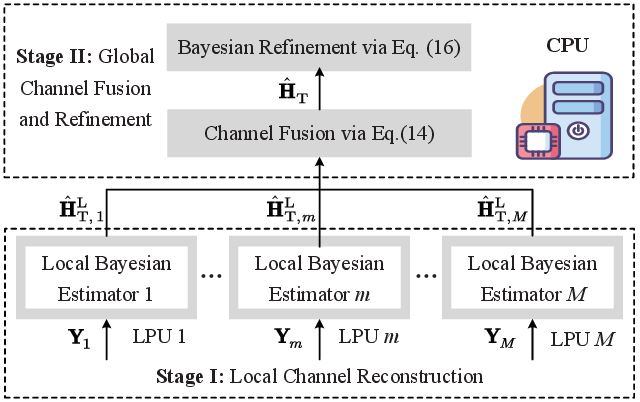}
		\caption{Proposed two-stage channel estimation scheme.}
		\label{Scheme}
	\end{figure}
	
	According to the model in (\ref{De_obs}), the local channel estimation for subarray $m$ can be derived as a sparse signal recovery problem based on the Bayesian inference framework, i.e.,
	\begin{equation}
		\hat{\mathbf{x}}_m \propto \arg \max_{\mathbf{x}_m}\quad p(\mathbf{y}_m|\mathbf{x}_m, \zeta_m)p(\mathbf{x}_m)p(\zeta_m), m \in \mathcal{M},
		\label{MAP1}
	\end{equation} 
	where $p(\mathbf{y}_m|\mathbf{x}_m, \zeta_m)$ and ${p(\mathbf{x}_m)}$ denote the joint likelihood function of $\mathbf{x}_m$ and $\zeta_m$ and the prior probability distribution of $\mathbf{x}_m$, respectively. 
	Notably, the estimation performance of problem (\ref{MAP1}) is dependent on the accuracy of prior distribution $p(\mathbf{x}_m)$. 
	Meanwhile, due to the limited computational capability of LPUs, a computationally efficient estimation scheme is needed.
	To solve this problem, we proposed a SBL-GNNs algorithm, which will be elaborated on in the Sec. \ref{section4}. 
	\vspace{-1em}
	\subsection{Stage II: Global Channel Fusion with Refinement} 
	Once the angular-delay channels $\hat{\mathbf{x}}_m$ are obtained, LPUs upload the estimates $\hat{\mathbf{H}}_{\mathrm{T},m}^{\mathrm{L}}$, to the CPU node for further fusion and refinement, where $\hat{\mathbf{H}}_{\mathrm{T},m}^{\mathrm{L}}$ denotes the reshaped angular-delay channel matrix based on $\hat{\mathbf{x}}_m$.
	\subsubsection{\textbf{Channel Fusion in Angular Domain}}
	The purpose of the central fusion is to align the local estimate $\hat{\mathbf{H}}_{\mathrm{T},m}^{\mathrm{L}}$ into the global angular domain with enhanced spatial resolution, thereby facilitating to leverage the correlations of subarray channel in the angular-delay domain. Specifically, leveraging the linear property of the DFT transformation matrix, we have
	\begin{equation}
		\mathbf{H}_{\mathrm{T}} = \sum_{m=1}^{M} \mathbf{F}_{\mathrm{A},m}\mathbf{H}_m\mathbf{F}_\mathrm{D}^{\mathrm{H}} \triangleq \sum_{m=1}^{M} \mathbf{H}_{\mathrm{T},m} , m \in \mathcal{M},
		\label{AD_Domain2}
	\end{equation}
	where $\mathbf{H}_{\mathrm{T},m}$ denotes the angular-delay channels of subarray $m$ under the global angular transformation matrix $\mathbf{F}_{\mathrm{A}} = \left[\mathbf{F}_{\mathrm{A}, 1}, \mathbf{F}_{\mathrm{A}, 2}, \cdots, \mathbf{F}_{\mathrm{A}, M} \right]$. 
	Based on this, the aggregated channel in the CPU can be given by
	\begin{equation}
		\hat{\mathbf{H}}_{\mathrm{T}} = \sum_{m=1}^{M} \mathbf{F}_{\mathrm{A},m}\mathbf{F}_{\mathrm{L}}^{\mathrm{H}}\hat{\mathbf{H}}_{\mathrm{T},m}^{\mathrm{L}}.
		\label{AWGN_Model2}
	\end{equation}
	
	\begin{remark}
		Compared to the local angular-delay channels $\mathbf{H}_{\mathrm{T},m}^{\mathrm{L}}$, the global angular-delay channel $\mathbf{H}_{\mathrm{T}}$ can more efficiently capture the inherent correlations among subarrays. Specifically, as the subarrays are spatially close to one another, $\mathbf{H}_{\mathrm{T},m}$ shares common angular and delay sparse supports, and exhibits high correlation. By effectively exploiting these common sparse supports and correlations, the reconstruction performance can be significantly enhanced.
	\end{remark}
	\subsubsection{\textbf{Bayesian Refinement}}
	{Utilizing the aggregated $\hat{\mathbf{H}}_{\mathrm{T}}$, the central refinement task can be derived as a denoising problem based on the following AWGN observation model, i.e.,}
	\begin{equation}
		\mathbf{r} = \mathbf{x} + \mathbf{z},
		\label{AWGN_Model}
	\end{equation}
	where $\mathbf{r} = \mathrm{vec}(\hat{\mathbf{H}}_{\mathrm{T}}) \in \mathbb{C}^{N_{{t}}K}$, $\mathbf{x} = \mathrm{vec}({\mathbf{H}}_{\mathrm{T}}) \in \mathbb{C}^{N_{{t}}K}$; 
	$\mathbf{z}$ denotes the observation noise with each entry $z_j$, $j \in \mathcal{J}$, obeying $\mathcal{CN}(z_j; 0, \kappa^{-1})$ and $\kappa$ indicating the unknown precision parameter of observation noise.
	 
	The aim of central refinement is to enhance the reconstruction performance by exploiting the sparsity and correlation of the global angular-delay channel. Similar to the local estimation of $\mathbf{x}_m$, the central refinement can be derived as
	\begin{equation}
		\tilde{\mathbf{x}} = \arg \max_{\mathbf{x}} \quad p(\mathbf{r}|\mathbf{x}, \kappa)p(\mathbf{x})p(\kappa),
		\label{MAP2}
	\end{equation}
	where $p(\mathbf{r}|\mathbf{x}, \kappa)$ denotes the joint likelihood function of $\mathbf{x}$ and $\kappa$; $p(\mathbf{x})$ and $p(\kappa)\propto \kappa^{-1}$ denote the prior probability distribution of $\mathbf{x}$ and $\kappa$, respectively.
	To solve problem (\ref{MAP2}) efficiently, we propose a variational message-passing-based denoising algorithm that leverages a Markov chain-based hierarchical prior to capture the sparsity and correlation of $\mathbf{x}$, as detailed in the Sec. \ref{section5}.
	Once the global angular-delay channel $\tilde{\mathbf{x}}$ is obtained, the global antenna-frequency-domain channel can be reconstructed as $\tilde{\mathbf{H}} = \mathbf{F}_{\mathrm{A}}^{\mathrm{F}}\tilde{\mathbf{H}}_{\mathrm{T}}\mathbf{F}_{\mathrm{D}}$ = $[\tilde{\mathbf{H}}_1^{\mathrm{T}}, \tilde{\mathbf{H}}_2^{\mathrm{T}}, \cdots, \tilde{\mathbf{H}}_M^{\mathrm{T}}]$. 
	Then, $\tilde{\mathbf{H}}_m
	$ are sent to the corresponding LPUs for subsequent algorithm designs.
	
	\begin{remark}
		Owing to its reconstruction-then-refinement strategy, the proposed two-stage estimation scheme achieves accuracy comparable to centralized methods while maintaining low computational complexity. Specifically, in the first stage, high-dimensional matrix inversion and matrix-vector multiplications are decomposed into several parallel low-dimensional operations, as detailed in Sec. \ref{section4}. In the second stage, only scalar operations are involved, as discussed in Sec. \ref{section5}. Consequently, the proposed approach significantly reduces computational complexity compared to traditional centralized estimation schemes.
		Moreover, the centralized refinement step effectively exploits correlations across subchannels from different LPUs, further enhancing estimation performance and bringing it closer to that of centralized estimation.
	\end{remark}
	
	\section{Proposed Local Channel Construction Scheme}
	\label{section4}
	In this section, we first provide a brief overview of the conventional SBL scheme and then introduce the proposed SBL-GNNs estimator.
	\subsection{SBL Estimator}
	It is assumed that the elements in $\mathbf{x}_m$ are independent and the following hierarchical sparsity-promoting prior is used
	\begin{align}
		\label{p_xga} p(\mathbf{x}_m|\boldsymbol{\gamma}_m) &= \prod_{n\in \mathcal{V}}\mathcal{CN}(x_{m,n}; 0, \gamma_{m,n}^{-1}), m \in \mathcal{M}\\
		\label{p_ga} p(\boldsymbol{\gamma}_m) &= \prod_{n\in \mathcal{V}} p(\gamma_{m,n}; \boldsymbol{\lambda}_{m,n}), m \in \mathcal{M},
	\end{align}
	 where $x_{m,n}$ and $\gamma_{m,n}$ are the $n$-th elements of $\mathbf{x}_m$ and $\boldsymbol{\gamma}_m$ corresponding to subarray $m$ with $n\in \mathcal{V}$, and $\gamma_{m,n}^{-1}$ denote the prior variance. In addition, $\boldsymbol{\lambda}_{m,n}$ denote the prior distribution parameters of ${\gamma}_{m,n}$. Utilizing (\ref{p_xga}), (\ref{p_ga}) and Expectation-Maximization (EM) framework, problem (\ref{MAP1}) can be efficiently solved by SBL-based algorithms \cite{Variance_State_1, Variance_State_2}, and the update steps in each iteration are as follows:
	
	\textbf{Observation Module (E-Step):}
	\begin{align}
		\label{update1} \boldsymbol{\Sigma}_m^{(t)}  &= \left(\zeta^{(t-1)}_m\boldsymbol{\Phi}_m^{\mathrm{H}}\boldsymbol{\Phi}_m + \mathrm{diag}\left(\boldsymbol{\gamma}_m^{(t-1)}\right)\right)^{-1},\\
		\label{update2} \boldsymbol{\mu}_m^{(t)} &=\zeta_m^{(t-1)} \boldsymbol{\Sigma}_m^{(t)}\boldsymbol{\Phi}_m^{\mathrm{H}}\mathbf{y}_m.
	\end{align}
	
	\textbf{Prior Update Module (M-Step):}
	\begin{align}
		\label{update3} \boldsymbol{\gamma}_m^{(t)} &= g_{1}\left(\boldsymbol{\mu}_m^{(t-1)}, \boldsymbol{\Sigma}_m^{(t-1)}, \boldsymbol{\lambda}_{m,n}^{(t-1)}\right),\\
		\label{update4} \boldsymbol{\lambda}_{m,n}^{(t)} &= g_{2}\left(\boldsymbol{\mu}_m^{(t-1)},\boldsymbol{\Sigma}_m^{(t-1)}, \boldsymbol{\gamma}_m^{(t)}\right),\\
		\label{update5} \zeta^{(t)}_m &= g_{3}\left(\boldsymbol{\mu}_m^{(t-1)},\boldsymbol{\Sigma}_m^{(t-1)}, \mathbf{y}_m\right),
	\end{align}
	where $\boldsymbol{\Sigma}_m^{(t)}$ and $\boldsymbol{\mu}_m^{(t)}$ represent the posterior covariance matrix and mean of $\mathbf{x}_m$ at the $t$-th iteration, respectively; {$g_1(\cdot)$, $g_2(\cdot)$, and $g_3(\cdot)$ denote the estimators for the prior variance, the hyperparameters of the prior variance, and the noise precision parameters, respectively. Specifically, $g_1(\cdot)$ leverages the $(t-1)$-th estimation results from the observation module and the $(t-1)$-th hyperparameters of the prior variance to update the $t$-th variance estimation. Subsequently, using the $(t-1)$-th estimation results from the observation module along with the current variance estimation, the hyperparameters of the prior variance are updated. Finally, the noise precision parameters are refined based on the inputs from the observation module and the observed signal. These updates are performed iteratively until convergence, at which point the posterior mean $\boldsymbol{\mu}_m^{(t)}$ is adopted as the final estimate of $\mathbf{x}_m$.}
	
	{Recently, various approaches have been proposed for designing $g_1(\cdot)$, $g_2(\cdot)$, and $g_3(\cdot)$, including mean-field-based methods \cite{Variance_State_1} and variational inference techniques \cite{Variance_State_2}. However, conventional sparse prior models often fail to fully leverage the correlations among channel entries. For instance, the Bernoulli-Gaussian prior \cite{StdSBL} and two-layer nested Gaussian-Gamma prior \cite{UAMP} assume independence among channel entries, which leads to performance degradation when addressing block-sparse signals. 
	Although advanced approaches such as the Markov chain-based and the 4-connected MRF-based hierarchical priors \cite{Anzheng2,Centralized_CE_6,Variance_State_1,Variance_State_2} have been developed, these methods only consider correlations among adjacent entries. This limitation renders them inadequate for effectively capturing the global correlations of local angular-delay channels.}

	\subsection{Proposed SBL-GNNs Estimator}
	To overcome the  limitations of conventional sparse prior models, we propose a GNN-based approach within the SBL framework to estimate the prior variance of the angular-delay channel and its associated parameters. Specifically, the sparse prior of $\mathbf{x}_m$ is modeled as a fully connected MRF. Leveraging the structural properties of the pairwise MRF, a GNN architecture is designed to replace the non-linear mapping functions $g_1(\cdot)$ and $g_2(\cdot)$ in the M-step are replaced with GNNs, thereby efficiently capturing the dependencies among channel coefficients. Consequently, the proposed SBL-GNNs algorithm can be regarded as a deep-unfolding counterpart of the traditional SBL algorithm with a fully connected MRF-based sparse prior.
	
	According to SBL principles, the output of the observation module in the $t$-th iteration can be equivalently characterized as an AWGN model, i.e.,
	\begin{equation}
		\boldsymbol{\mu}_m^{(t)} = \mathbf{x}_m + \mathbf{w}_m^{(t)}, m \in \mathcal{M},
		\label{AWGN2}
	\end{equation}
	where $\mathbf{w}_m^{(t)}$ denotes the equivalent residual noise in the $t$-th iteration with covariance matrix $\boldsymbol{\Sigma}_m^{(t)}$. 
	To fully capture the global correlations of local angular-delay channels, consider a MRF-based hierarchical prior given by
	\begin{align}
		p(\mathbf{x}_m, \boldsymbol{\gamma}_m; \mathbf{s}_m)&={p(\mathbf{x}_m| \boldsymbol{\gamma}_m)}{p(\boldsymbol{\gamma}_m| \mathbf{s}_m)}{p(\mathbf{s}_m)},
		\label{prior}
	\end{align}
	where the first layer is assigned a conditional
	Gaussian prior $p(\mathbf{x}_m| \boldsymbol{\gamma}_m)= \mathcal{CN}\left(\mathbf{x}_m; \mathbf{0}, \mathbf{R}_m\right)$ with $\mathbf{R}_m = \mathrm{diag}\left(\boldsymbol{\gamma}_m^{-1}\right) \in \mathbb{R}^{NK\times NK}$. {The second layer $p(\boldsymbol{\gamma}_m|\mathbf{s}_m)$ is modeled as a conditional Gamma distribution given by
	\begin{align}
		\label{2nd_layer} p(\boldsymbol{\gamma}_m|\mathbf{s}_{m})&=\prod_{n\in \mathcal{V}} p(\gamma_{m,n}|{s}_{m,n}), m \in \mathcal{M}, \\
		p(\gamma_{m,n}|{s}_{m,n}) &= \mathcal{G}a(\gamma_{m,n}; a_m,b_m)\delta(1-s_{m,n})\\
		\notag &+\delta(\gamma_{m,n})\delta(1+s_{m,n}), m \in \mathcal{M}, n \in \mathcal{V},
	\end{align}
	where $s_{m,n} \in \left\{-1,1\right\}$ denotes the $n$-th element of $\mathbf{s}_m$ is a hidden state variable of variance precision parameter $\gamma_{m,n}$ with $m \in \mathcal{M}$ and $n \in \mathcal{V}$;}
	$\delta(\cdot)$ denotes the Delta function;
	$a_m$ and $b_m$ denote the shape parameters of Gamma distribution. 
	Finally, the third layer is assigned a MRF given by
	\begin{equation}
		p(\mathbf{s}_m) = \left(\prod_{n \in \mathcal{V}}\prod_{k\neq n}\psi(s_{m,n}, s_{m,k})\right)^{\frac{1}{2}}\prod_{n \in \mathcal{V}}\phi(s_{m,n}),
		\label{3rd_layer}
	\end{equation}
	{where $\psi(s_{m,n}, s_{m,k}) = \exp(\beta_m s_{m,n} s_{m,k})$ and $\psi(s_{m,n}) =  \exp(-\alpha_m s_{m,n})$, $m\in \mathcal{M}$, $n,k\in \mathcal{V}$ denote the pair-potential and self-potential functions with $\alpha_m$ and $\beta_m$ as hyperparameters of MRF.}
	
	\begin{lemma}
		The joint posterior probability distribution $p(\mathbf{x}_m,\boldsymbol{\gamma}_m,\mathbf{s}_m|\boldsymbol{\mu}_m^{(t)})$ can be factorized as a pair-wise MRF, which is given by (\ref{MRF}), as shown in the top of this page,
		where $\mathbf{v}_{m,n}^{(t)}$ and $v^{n,k}_m(t)$ indicate the $n$-th column and $(n,k)$-th entry of $\mathbf{V}_m^{(t)} = (\boldsymbol{\Sigma}_m^{(t)})^{-1}$, respectively; $\psi_{\mathrm{post}}(x_{m,n})$ and $\phi_{\mathrm{post}}(x_{m,n}, x_{m,k})$ denote the self potential of $x_{m,n}$, $\forall n \in 
		\mathcal{V}$ and pair-potential functions between $x_{m,n}$ and $x_{m,k}$, $\forall n,k \in \mathcal{V}$ and $k\neq n$, respectively.
		\begin{figure*}[!t]
			\normalsize
			\setcounter{MYtempeqncnt}{\value{equation}}
			\setcounter{equation}{28}
			\begin{equation}
				\begin{aligned}
					p\left(\mathbf{x}_m,\boldsymbol{\gamma}_m,\mathbf{s}_m|\boldsymbol{\mu}_m^{(t)}\right) &\propto 
					\prod_{n\in \mathcal{V}}\psi_{\mathrm{post}}\left(x_{m,n}\right)\prod_{n\in \mathcal{V}}\prod_{k\neq n}\phi_{\mathrm{post}}\left(x_{m,n}, x_{m,k}\right), m \in \mathcal{M}, n,k \in \mathcal{V},\\
					\psi_{\mathrm{post}}\left(x_{m,n}\right) &\triangleq \psi_{\mathrm{like}}\left((\boldsymbol{\lambda}_m^{(t)})^{\mathrm{H}}\mathbf{v}_{m,n}^{(t)}, v^{n,n}_m(t),x_{m,n}\right)\psi_{\mathrm{pri}}\left(\alpha_m, x_{m,n}\right),\\
					\phi_{\mathrm{post}}\left(x_{m,n}, x_{m,k}\right) &\triangleq\phi_{\mathrm{like}}\left(v^{n,k}_m(t), x_{m,n}, x_{m,k}\right) \phi_{\mathrm{pri}}\left(\beta_m, x_{m,n}, x_{m,k}\right).
				\end{aligned}
				\label{MRF}
			\end{equation}
			\hrulefill
			\vspace*{4pt}
		\end{figure*}
	\end{lemma}
	\begin{proof}
		See Appendix A.
	\end{proof}
	
 Utilizing the \textbf{Lemma 1}, the precision vector $\boldsymbol{\gamma}_m^{(t)}$ is learned by maximizing  a posteriori  probability given by
	\begin{equation}
		\boldsymbol{\gamma}_m^{(t+1)} = \arg \max_{\boldsymbol{\gamma}_m} \quad p\left(\boldsymbol{\gamma}_m|\boldsymbol{\mu}_m^{(t)}\right), m \in \mathcal{M},
	\end{equation}
	where $p(\boldsymbol{\gamma}_m|\boldsymbol{\mu}_m^{(t)})$ denotes the marginal posteriori probability of $\boldsymbol{\gamma}_m$ given by
	\begin{equation}
		p\left(\boldsymbol{\gamma}_m|\boldsymbol{\mu}_m^{(t)}\right) = \int p(\mathbf{x}_m,\boldsymbol{\gamma}_m,\mathbf{s}_m|\boldsymbol{\mu}_m^{(t)}) \mathrm{d}\mathbf{x}_m \mathrm{d}\mathbf{s}_m.
		\label{posterior_ga}
	\end{equation}
	However, the introduction of the MRF couples the variables $x_{m,n}$, making the integral in (\ref{posterior_ga}) computationally intractable. To address this challenge, we propose leveraging  GNNs to learn the prior variance parameters. Specifically, using the pairwise MRF model provided in (\ref{MRF}), the dependency among the angular-delay channel coefficients $x_{m,n}$ for $\forall m \in \mathcal{M}$ and $\forall n \in \mathcal{V}$ can be represented by an undirected graph $G^m = \{V^m, E^m\}$. 
	Here, $V^m = \{x_{m,1}, x_{m,2}, \dots, x_{m,NK}\}$ denotes the set of nodes, and $E^m = \{e^m_{n,k}\}$, where $n, k \in \mathcal{V}$ and $k \neq n$, represents the set of edges capturing the dependencies. The prior variance $\boldsymbol{\gamma}_m$ can then be interpreted as a hidden attribute of $G^m$ and effectively learned using GNNs applied to $G^m$.

	According to the above analysis, the block diagram of the proposed SBL-GNNs is illustrated in Fig. \ref{SBL_GNN}, where the network consists of $T$ cascade layers. Each layer shares the same structure, comprising GNNs module and an observation module. 
	The input of the SBL-GNNs network is the received signal of each subarray $\mathbf{y}_m$ and the noise level $\zeta_m$, with the initial $\boldsymbol{\gamma}_m^{(0)} = \mathbf{1}_{NK}$ and $\zeta_m^{(0)} = 1$, and the output is the posterior mean $\boldsymbol{\mu}_m^{(T)}$. 
	For any middle layer $t$ with $1<t<T$, its input is the received signal $\mathbf{y}_m$, the estimated variance $\boldsymbol{\gamma}_m^{(t-1)}$ and noise level $\zeta_m^{(t-1)}$ from $t-1$ layer. 
	Similar to the SBL-based algorithms \cite{Variance_State_1, Variance_State_2}, the SBL-GNNs network is sequentially executed until termination by a fixed number of layers. 
	\begin{figure*}
		\centering
		\includegraphics[width=0.7\textwidth]{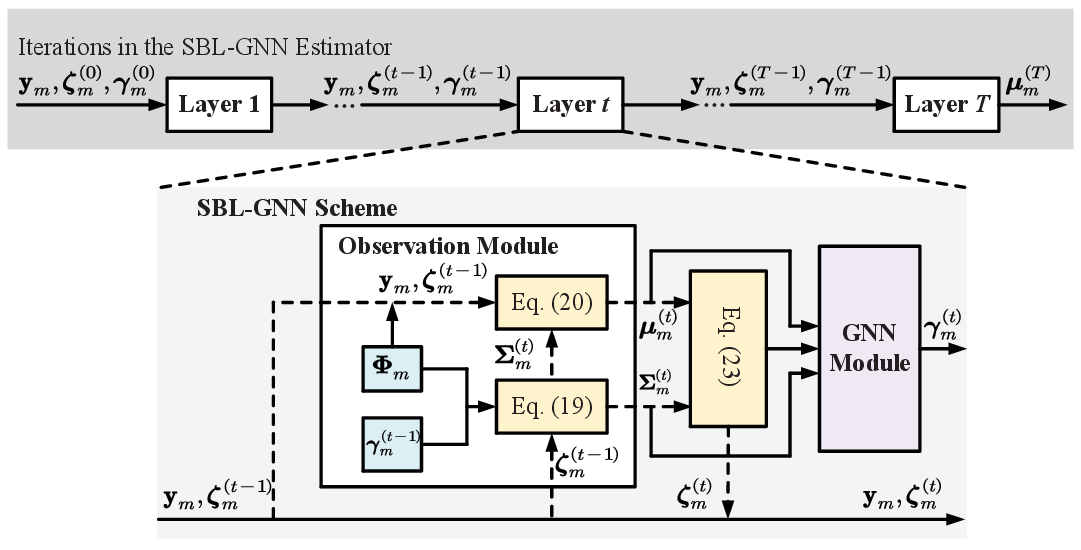}
		\caption{The structure of the proposed SBL-GNNs network}
		\label{SBL_GNN}
	\end{figure*}
	
	For notation convenience, we omit the index $t$ in subsequent section. For GNNs module, the output message from observation module can be seen as the prior knowledge for the variable node $x_{m,n}$, which is given by a Gaussian distribution with mean $\boldsymbol{\lambda}_m$ and covariance matrix $\boldsymbol{\Sigma}_m$. Additionally, the self potential of each node is associated with $\alpha_m$. Therefore, the node attribute $\mathbf{a}^{(l)}_{m,n} \in \mathbb{R}^{4\times 1}$ is yielded by 
	\begin{equation}
		\mathbf{a}_{m,n}^{(l)} = \left[\boldsymbol{\lambda}_m^{\mathrm{H}}\mathbf{v}_{m,n}, v_m^{n,n}, \alpha_m , \zeta_m \right]^{\mathrm{T}}, m\in \mathcal{M}, n\in \mathcal{V},
	\end{equation}
	where $l$ denotes the index of the number of iterations in the GNNs modules.
	The attribute $\mathbf{f}_m^{n,k} \triangleq \left[v^{n,k}_m, \beta_m, \zeta_m\right]^{\mathrm{T}} \in \mathbb{R}^{3\times 1}$ of edge $e_{m,k}^m$ is obtained by extracting the information from the pair-potential function in (\ref{MRF}). 
	Additionally, each variable node associates with a hidden vector $\mathbf{u}^{(l)}_{m,n}$, which is determined by the node attribute $\mathbf{a}^{(l)}_{m,n}$ and the messages from other nodes. 
	In particular, for the first iteration, the hidden vector $\mathbf{u}^{(0)}_{m,n}$ is given by
		\begin{equation}
			\mathbf{u}^{(0)}_{m,n} = \mathbf{W}_1\mathbf{a}_{m,n} + \mathbf{b}_1, m\in \mathcal{M}, n\in \mathcal{V},
			\label{u0}
		\end{equation}
	where $\mathbf{W}_1 \in \mathbb{R}^{N_{\mathrm{u}}\times 4}$ and $\mathbf{b}_1 \in \mathbb{R}^{N_{\mathrm{u}}\times 1}$ denote the learnable weight matrix and bias with $N_{\mathrm{u}}$ indicating the output message dimension of the variable node.
	In the following, we elaborate on the propagation, aggregation, and readout sub-modules of GNNs module.  
	\subsubsection{Propagation module} In propagation module, for any variable nodes $x_{m,n}$  and $x_{m,k}$ with $n\neq k$, there exists a factor node, which concentrates the edge attributes and the messages $\mathbf{u}^{(l)}_{m,n}$ and $\mathbf{u}^{(l)}_{m,k}$ from $x_{m,n}$ and $x_{m,k}$ and is given by
	\begin{equation}
		\mathbf{c}_m^{n,k}(l) = \mathcal{P}\left(\mathbf{u}^{(l-1)}_{m,n}, \mathbf{u}^{(l-1)}_{m,k}, \mathbf{f}_m^{n,k}\right), m\in \mathcal{M}, n,k \in \mathcal{V}, 
		\label{message_ij}
	\end{equation}
	where $\mathcal{P}(\cdot)$ denotes a multi-layer perceptron (MLP) with two hidden layers of sizes $N_{\mathrm{h}_1}$ and $N_{\mathrm{h}_2}$ and an output layer of size $N_{\mathrm{u}}$. 
	Furthermore, the rectifier linear unit (ReLU) activation function is used at the	output of each hidden layer. 
	Finally, the outputs $\mathbf{c}_m^{n,k}(l)$ are fed back to neighboring variable nodes. 
	\subsubsection{Aggregation module}In each variable node $x_{m,n}$, all messages $\mathbf{c}_m^{n,k}(l) \in \mathbb{R}^{N_{\mathrm{u}}\times 1}$ associated with connected edges $e^{n,k}_m$ are first added, and then the sum and the node attribute are concentrated, i.e., 
	\begin{equation}
		\mathbf{c}_m^{n}(l) = \left[\sum_{k\neq n}\left(\mathbf{c}_m^{n,k}(l)\right)^{\mathrm{T}}, \mathbf{a}_{m,n}^{\mathrm{T}} \right]^{\mathrm{T}},
		\label{message_i}
	\end{equation}
	Utilizing $\mathbf{c}_m^{n}(l)$, the hidden vector for each variable node can be updated according to 
	\begin{align}
		\label{hidden1} \mathbf{g}_{m,n}^{(l)} &= \mathcal{R}\left(\mathbf{g}_{m,n}^{(l-1)}, \mathbf{c}_m^{n}(l)\right),\\
		\label{hidden2} \mathbf{u}_{m,n}^{(l)} &= \mathbf{W}_2\mathbf{g}^{(l)}_{m,n} + \mathbf{b}_2,
	\end{align}
	where $\mathcal{R}(\cdot)$ denote the the gated recurrent unit (GRU) network, whose current and previous hidden states are $\mathbf{g}_{m,n}^{(l)}$ and $\mathbf{g}_{m,n}^{(l-1)}$, respectively. $\mathbf{W}_2 \in \mathbb{R}^{N_{\mathrm{u}}\times N_{\mathrm{h}_1}}$ and $\mathbf{b}_2 \in \mathbb{R}^{N_{\mathrm{u}}\times 1}$ denote the learnable weight matrix and bias. 
	The updated feature vector $\mathbf{u}_{m,n}^{(l)}$ is then sent to the propagation module for the next iteration.
	
	\subsubsection{Readout module}After $L$ rounds of the message passing
	between the propagation and aggregation modules, a readout
	module is utilized to output the estimated result, i.e., 
	\begin{equation}
		{\gamma}_{m,n}^{(t)} = \mathcal{H}(\mathbf{u}_{m,n}^{(L)}),
		\label{prior_variance}
	\end{equation}
	where $\mathcal{H}(\cdot)$ denotes the MLP with two hidden layers of sizes $N_{\mathrm{h}_1}$ and $N_{\mathrm{h}_2}$ and an output layer of size $1$, where the ReLU activation is utilized as the output of the last layer. 
	
	\begin{remark}
		Due to the differences in estimation and prior distribution, the pair-wise MRF representation exhibits significant differences with existing works in \cite{GNN1, GNN2, GNN3}, resulting in distinct graph attributes and GNN architectures. Thus, we reconsider the graph attributes, including the pair-wise MRF representation (see (\ref{MRF})) as well as the node and edge attributes tailored to the channel estimation problem.
		Meanwhile, from a system perspective, the proposed GNN is designed to capture the correlations across different channel coefficients through the fully connected MRF, whereas the GNNs in \cite{GNN1,GNN2,GNN3} primarily address correlated and non-Gaussian residual noise. 
		Additionally, due to the differences in the types of parameters to be estimated, the readout module of GNNs differ fundamentally from existing works. Specifically, the readout module designed for discrete variables in MIMO detection tasks is not applicable to channel estimation. To this end, we redesigned the readout module to accommodate the regression tasks associated with continuous variables in the channel estimation problem.
	\end{remark}
	\vspace{-1em}
	\subsection{The Overall Algorithm Description}
	\begin{algorithm}
		\renewcommand{\algorithmicrequire}{\textbf{Input:}}
		\renewcommand{\algorithmicensure}{\textbf{Output:}}
		\caption{Proposed SBL-GNNs algorithm}
		\begin{algorithmic}[1]
			\Require $\mathbf{y}_m$, $\boldsymbol{\Phi}_m$, $\mathbf{F}_{\mathrm{L}}$, $\mathbf{F}_{\mathrm{D}}$, $\zeta$, $T$, and $L$.
			\Statex \textbf{Initialize:} ${\boldsymbol{\gamma}}_m^{(0)}=\mathbf{1}$ and $\mathbf{g}_{m,i}^{(l-1)}=\mathbf{0}$;
			\For {$t = 1, \cdots, t, \cdots T$}
			\Statex /*\textbf{Observation module}*/
			\State Compute $\boldsymbol{\Sigma}_m^{(t)}$ and $\boldsymbol{\lambda}_m^{(t)}$ in (\ref{update3}) and (\ref{update4});
			\State Compute $\mathbf{V}_m^{(t)} = \zeta_m\boldsymbol{\Phi}_m^{\mathrm{H}}\boldsymbol{\Phi}_m + \mathrm{diag}(\boldsymbol{\gamma}_m^{(t-1)})$;
			\Statex /*\textbf{GNNs module}*/
			\If{$t$=1}
			\State Compute $\mathbf{u}^{(0)}_{m,n}$ in (\ref{u0}), $n\in \mathcal{V}$;
			\EndIf
			\For{$l=1,\cdots,l,\cdots L$}
			\State Compute $\mathbf{c}_m^{n,k}(l)$ in (\ref{message_ij}),  $n, k \in \mathcal{V}$ and $k\neq n$;
			\State Compute $\mathbf{c}_m^{n}(l)$ in (\ref{message_i}),  $n\in \mathcal{V}$;
			\State Compute $\mathbf{g}_{m,n}^{(l)}$ in (\ref{hidden1}), $n\in \mathcal{V}$;
			\State Compute $\mathbf{u}_{m,i}^{(l)}$ in (\ref{hidden2}), $n\in \mathcal{V}$;
			\EndFor
			\State Compute the prior variance ${\gamma}_{m,n}^{(t)}$ in (\ref{prior_variance});
			\EndFor
			\Ensure  $\hat{\mathbf{x}}_m = \boldsymbol{\lambda}_m^{(T)}$.
		\end{algorithmic}
		\label{GNN_SBL_Algorithm}
	\end{algorithm}
	
	The proposed SBL-GNNs algorithm is composed of two modules, i.e., observation module and GNNs module. The computations in the observation module remain
	the same as those in the traditional SBL-based algorithm. The difference between SBL-GNNs and traditional SBL is that the update of prior variance is replaced by GNNs. {The overall SBL-GNNs algorithm is given in Algorithm \ref{GNN_SBL_Algorithm}.}
	In the following, we elaborate on the computational complexity in terms of the number of multiplications of the proposed SBL-GNNs algorithm. In the observation module, we first need to compute $\boldsymbol{\Phi}_m^{\mathrm{H}}\boldsymbol{\Phi}_m$ and $\boldsymbol{\Phi}_m^{\mathrm{H}}\mathbf{y}_m$. 
	The computational complexity is $\mathcal{O}(K^3N^2PN_{\mathrm{RF}})$ and $\mathcal{O}(K^2PN_{\mathrm{RF}}N)$, respectively. Then, the matrix inversion operation in (\ref{update3}) involves the computational complexity of $\mathcal{O}(N^3K^3+N^2K^2+2NK)$. In the GNNs module, the main computational complexity is from the MLP, whose complexity is $\mathcal{O}(LN_{\mathrm{u}}N_{\mathrm{h}_1}N_{\mathrm{h}_2})$. Thus, the over computational complexity is $\mathcal{O}(K^3N^2PN_{\mathrm{RF}}T+K^2PN_{\mathrm{RF}}NT+N^3K^3T+N^2K^2T+2NKT+LN_{\mathrm{u}}N_{\mathrm{h}_1}N_{\mathrm{h}_2}T)$.
	\vspace{-1em}
	\subsection{Training Scheme}
	The proposed SBL-GNNs estimator was implemented in Python using PyTorch. The training process consisted of 50 epochs, with each epoch comprising 320 batches of 64 samples. Each sample included realizations of $\mathbf{x}_m$, $\boldsymbol{\Phi}_m$, $\mathbf{y}m$, and $\mathbf{n}m$ generated according to (\ref{De_obs}). The signal-to-noise ratio (SNR) values were uniformly distributed in the range $[SNR{\mathrm{min}}, SNR{\mathrm{max}}]$. The Adam optimizer was employed for training with a learning rate of $1 \times 10^{-4}$. A mean squared error (MSE) loss function was utilized to update the learnable parameters via the back-propagation algorithm.

	In the testing phase, a dataset comprising 10,000 samples per SNR point was generated for evaluation. The normalized mean squared error (NMSE) was computed for all trained estimators using this testing dataset. It is worth noting that the proposed estimators share the same neural network (NN) parameters across all subarrays. This design ensures that the number of NN parameters to be optimized during training remains independent of the number of subarrays, thereby maintaining scalability.

	\section{Proposed Centralized Bayesian Refinement Algorithm}
	\label{section5}
	In this section, we detail the proposed Bayesian refinement algorithm to effectively solve the denoising problem in (\ref{MAP2}). Specifically, to capture both the block sparsity and the correlations among the angular-delay channel coefficients ${x}_j$, we assign a Markov chain-based hierarchical prior model for $\mathbf{x}$, which is given by
	\begin{equation}
		\begin{aligned}
			p(\mathbf{x}, \boldsymbol{\upsilon}, \mathbf{o}) &=p(\mathbf{x}|\boldsymbol{\upsilon})p(\boldsymbol{\upsilon}|\mathbf{o})p(\mathbf{o})\\
			&=\prod_{j\in\mathcal{J}}p(x_j|\upsilon_j)p(\upsilon_j|o_j)p(\mathbf{o}),
		\end{aligned}
		\label{Prior}
	\end{equation}
	where $p(x_j|\upsilon_j) = \mathcal{CN}(x_j; 0, \upsilon_j^{-1})$, $p(\upsilon_j|o_j) = \delta(1-o_{j})\mathcal{G}a(\upsilon_{j};a,b) + \delta(1+o_{j})\mathcal{G}a(\upsilon_{j};\overline{a},\overline{b})$, and $p(\mathbf{o})$ denotes a Markov chain given by $p(\mathbf{o})=p(o_1)\prod_{j=2}^{N_t}p(o_j|o_{j-1})$, where $p(o_j|o_{j-1})$ is given by
	\begin{equation}
		p(o_j|o_{j-1}) = \left\{
		\begin{aligned}
			(1-p_{01})^{1-o_j}p_{01}^{o_j}&,\quad o_{j-1} = 0, \\
			p_{10}^{1-o_j} (1-p_{10})^{o_j}&,\quad o_{j-1} = 1,\\
		\end{aligned}
		\right.
		\label{trans}
	\end{equation}
	where $p_{01} = p(o_j = 0|o_{j-1} = 1)$ and the other three transition probabilities are defined similarly. 
	
	According to (\ref{Prior}) and (\ref{trans}), the prior of $x_j$ is given by
		\begin{equation}
			p(x_j) = \int\mathcal{CN}(x_j; 0, \upsilon_j^{-1})p(\upsilon_j|o_j)p(\mathbf{o})\mathrm{d}\upsilon_j\mathrm{d}\mathbf{o}, j \in \mathcal{J}.
			\label{p4}
	\end{equation}
	From (\ref{p4}), the Markov chain-based prior model facilitates a flexible and effective representation of sparsity and correlations among the coefficients $x_j$. Specifically, the precision parameters of the variance, $\upsilon_j$, regulate the sparsity of $x_j$, where $\upsilon_j$ approaching zero effectively enforces $x_j$ to be zero. Furthermore, the Markov chain models the variance states, encouraging structured block patterns while suppressing isolated coefficients that deviate from their neighboring values.
	Precisely, the state of each coefficient being non-zero is conditioned on the states of its adjacent coefficients, with transition probabilities $p(o_j|o_{j-1})$ dictating the likelihood of transitions between states. 
	When a coefficient is in the non-zero state, the transition probability of remaining in this state is expected to be higher, thereby promoting continuity and reducing the occurrence of isolated non-zero coefficients.
	
	\begin{figure*}
		\centering
		\includegraphics[width=0.7\textwidth]{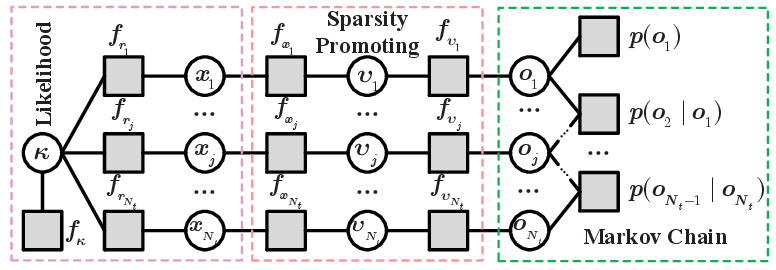}
		\caption{The factor graph representation of (\ref{MAP3}).}
		\label{Centralized}
	\end{figure*}
	 
	With the prior model in (\ref{Prior}), problem (\ref{MAP2}) can be formulated as 
	\begin{equation}
		\begin{aligned}
			\tilde{\mathbf{x}} = \arg \max_{\mathbf{x}} \quad p(\mathbf{r}|\mathbf{x},\kappa)p(\mathbf{x}| \boldsymbol{\upsilon})p(\boldsymbol{\upsilon}|\mathbf{o})p(\mathbf{o})p(\kappa),
		\end{aligned}
		\label{MAP3}
	\end{equation}
	where $p(\mathbf{r}|\mathbf{x}, \kappa)$ denotes the likelihood probability distribution, which is given by
	\begin{equation}
		p(\mathbf{r}|\mathbf{x}, \kappa) \propto \frac{1}{\pi^{^{N_t}}\kappa^{-^{N_t}}} \exp(-\kappa\left|\mathbf{x}-\mathbf{r}\right|^2).
	\end{equation}
	The dependencies of the random variables in the factorization (\ref{MAP3}) can be shown by a factor graph as depicted in Fig.~\ref{Centralized}, where $f_{r_j}\triangleq p(r_j\mid x_j,\kappa)$, $f_{x_j} \triangleq p(x_j \mid \upsilon_j)$, $f_{\upsilon_j} \triangleq p(\upsilon_j \mid o_j)$, and $f_{\kappa} \triangleq p(\kappa)$.
	To achieve a computationally efficient inference of problem (\ref{MAP3}), we propose a variational message passing and belief propagation based refinement algorithm. 
	\begin{remark}
		From (\ref{p4}), it is evident that, unlike \cite{UAMP}, where each $x_{j}$ is assigned an independent variance precision parameter, our approach explicitly captures the correlations among angular-delay channel coefficients by the Markov chain. This fundamental distinction leads to significant differences in the factor graph structure. Specifically, the sparse prior models in \cite{UAMP} are typically formulated for loop-free factor graphs.
		In contrast, our approach introduces loops and higher-order dependencies through the Markov chain, resulting in a more intricate factor graph representation, as illustrated in Fig. \ref{Centralized}. This structural complexity necessitates distinct message update equations at the variable nodes, deviating from conventional message-passing strategies designed for \cite{UAMP}.
	\end{remark}
	
	In the following, we derive the forward (from left to right) and backward (from right to left) message passing in detail.
	\subsubsection{\textbf{Forward Message Passing}}Assume the belief of $\kappa$ is denoted as $b(\kappa)$, which is derived in (\ref{post_beta1}). Utilizing the variational message passing rule, the message from factor node $p(r_j|x_j)$ to $x_j$ can be given by
	\begin{equation}
		\begin{aligned}
			\nu_{f_{r_j}\rightarrow x_j}(x_j) &\propto \exp\left\{\int b(\kappa)\ln p(r_j|x_j,\kappa) \mathrm{d}\kappa\right\}\\ &\propto \mathcal{CN}(x_j; r_j, \tilde{\kappa}^{-1}).
		\end{aligned}
	\end{equation}
	
	Assume that the belief of $x_j$ obeys $\mathcal{CN}(x_j; \tilde{x}_j, \tilde{\tau}_j)$, which is derived in (\ref{post_x_q}). Utilizing the message $\nu_{f_{r_j}\rightarrow x_j}(x_j)$ and variational message passing rules, 
	the message from factor node $p(x_j|v_j)$ to variable node $\upsilon_j$ is given by
	\begin{equation}
		\begin{aligned}
			\nu_{f_{x_j}\rightarrow \upsilon_j}(\upsilon_j) &\propto \exp\left\{\int b(x_j)\ln p(x_j|\upsilon_j) \mathrm{d}x_j\right\} \\
			&\propto \upsilon_j\exp(-\upsilon_j(\left|\tilde{x}_j\right|^2+\tilde{\tau}_j)),
		\end{aligned}
	\end{equation}
	Then, according to the sum-product rules, the message from factor node $p(v_j|o_j)$ to $o_j$ is given by
	\begin{equation}
		\begin{aligned}
			\nu_{f_{\upsilon_j}\rightarrow o_j}(o_j)
			=\pi_j^{\mathrm{out}}\delta(1-o_j) + (1-\pi_j^{\mathrm{out}})\delta(1+o_j),
		\end{aligned}
		\label{fv2o}
	\end{equation}
	where $\pi_j^{\mathrm{out}}$ is given by
	\begin{equation}
		\pi_j^{\mathrm{out}} = \frac{a(\overline{b}+\left|\tilde{x}_j\right|^2+\tilde{\tau}_j)}{a(\overline{b}+\left|\tilde{x}_j\right|^2+\tilde{\tau}_j)+\overline{a}({b}+\left|\tilde{x}_j\right|^2+\tilde{\tau}_j)}.
		\label{pi_out}
	\end{equation}
	\subsubsection{\textbf{Backward Message Passing}}
	According to (\ref{fv2o}), the message from $o_j$ to $p(v_j|o_j)$ is given by ${\nu}_{o_j \rightarrow f_{\upsilon_j}}(o_j) \propto  (1-\pi^{\mathrm{in}}_j) \delta(o_j) + \pi^{\mathrm{in}}_j \delta(1-o_j)$, 
	where $\pi^{\mathrm{in}}_j$ is defined as 
	\begin{equation}
		\pi^{\mathrm{in}}_j = \frac{\psi^f_j \psi^b_j}{\psi^f_j \psi^b_j + (1-\psi^f_j)(1-\psi^b_j)}.
		\label{pi_in}
	\end{equation}
	Here, $\psi^f_j$ and $\psi^b_j$ are the forward and backward messages along the Markov chain, and are respectively defined as\cite{Anzheng2,Centralized_CE_6}
	\begin{equation}
		\label{psi_f} \psi^f_j = \frac{p_{01}(1-\psi^f_{j-1})(1-\pi_{j-1}^{\mathrm{out}})+ p_{11}\lambda^f_{j-1} \pi_{j-1}^{\mathrm{out}}}{(1-\psi^f_{j-1})(1-\pi_{j-1}^{\mathrm{out}})+ \psi^f_{j-1} \pi_{j-1}^{\mathrm{out}}},
	\end{equation}
	\begin{equation}
		\label{psi_b} \psi^b_j = \frac{p_{10}(1-\psi^b_{j+1})(1-\pi_{j+1}^{\mathrm{out}})+p_{11}\psi^b_{j+1}\pi_{j+1}^{\mathrm{out}}}{p_0(1-\psi^b_{j+1})(1-\pi_{j+1}^{\mathrm{out}})+p_1\psi^b_{j+1}\pi_{j+1}^{\mathrm{out}}}.
	\end{equation}
	where $p_0 = p_{10}+p_{00}$, $p_1 = p_{11}+p_{01}$.
	
	Then, the message from $p(\upsilon_j|o_j)$ to $\upsilon_j$ can be derived as
	\begin{equation}
		\begin{aligned}
			\nu_{f_{\upsilon_{j}}\rightarrow \upsilon_{j}}(\upsilon_{j}) &= \sum_{o_j \in \left\{-1,1\right\}} p(\upsilon_{j}|o_j) \nu_{o_j\rightarrow f_{\upsilon_j}}(o_j)\\
			&=\pi^{\mathrm{in}}_j\mathcal{G}a(\gamma; a, b) + (1-\pi^{\mathrm{in}}_j)\mathcal{G}a(\gamma; \overline{a}, \overline{b}).
		\end{aligned}
	\end{equation} 
	With the message $\nu_{f_{x_j}\rightarrow \upsilon_j}(\upsilon_j)$ and $\nu_{f_{\upsilon_j}\rightarrow \upsilon_j}(\upsilon_j)$, the belief of $\upsilon_j$ can be given by
	\begin{equation}
		\begin{aligned}
			b(\upsilon_j) &= \pi^{\mathrm{in}}_j\upsilon_j^a\exp(-\upsilon_j(b+\left|\hat{x}_n\right|^2+\tilde{\tau}_j))\\
			&+(1-\pi^{\mathrm{in}}_j)\upsilon_j^{\overline{a}}\exp(-\upsilon_j(\overline{b}+\left|\tilde{x}_j\right|^2+\tilde{\tau}_j))).
			\label{b_ga}
		\end{aligned}
	\end{equation}
	From (\ref{b_ga}), the approximate posterior mean of $\gamma_j$ is given by
	\begin{equation}
		\begin{aligned}
			\hat{\upsilon}_j = \pi^{\mathrm{in}}_j \frac{a+1}{b+\left|\hat{x}_j\right|^2+\tilde{\tau}_j}
			+(1-\pi^{\mathrm{in}}_j)\frac{\overline{a}+1}{\overline{b}+\left|\tilde{x}_j\right|^2+\tilde{\tau}_j}.
		\end{aligned}
		\label{hat_ga}
	\end{equation} 
	Consequently, the message from $p(x_j|\upsilon_j)$ to $x_j$ is derived as
	\begin{equation}
		\begin{aligned}
			\nu_{f_{x_j} \rightarrow x_j}(x_j) &\propto \exp\left\{\int\ln f_{x_j}(x_j, \upsilon_j) b(\upsilon_j)\mathrm{d} \upsilon_j \right\}\\
			&\propto\mathcal{CN}(x_j; 0, \tilde{\upsilon}_j^{-1}),
		\end{aligned}
		\label{m_fx2x}
	\end{equation}
	which will be utilized to update the approximate posterior distribution.
	
	Combining the messages $\nu_{f_{r_j}\rightarrow x_j}(x_j)$ and $\nu_{f_{x_j} \rightarrow x_j}(x_j)$,the approximate posterior marginal distribution of $x_j$ can be approximated as
	\begin{equation}
		\begin{aligned}
			b(x_j) \propto {\nu_{f_{r_j}\rightarrow x_j}(x_j)\nu_{f_{x_j} \rightarrow x_j}(x_j)}\propto \mathcal{CN}(x_j; \tilde{x}_j, \tilde{\tau}_j),
		\end{aligned}
		\label{post_x_q}
	\end{equation}
	 where the approximate posterior mean and variance of $x_j$ can be given by
	\begin{align}
		\label{hat_x} \tilde{x}_j = \frac{r_j\tilde{\kappa}}{\tilde{\kappa}+\tilde{\upsilon}_j}, \quad  \tilde{\tau}_j = \frac{1}{\tilde{\kappa}+\tilde{\upsilon}_j}.
	\end{align}
	
	Based on the belief $b(x_j)$ and variational message passing rules, the message from $f_{r_j}$ to $\kappa$ is given by
	\begin{equation}
		\begin{aligned}
			\nu_{f_{r_j} \rightarrow \kappa} (\kappa) &= \exp\left\{\int b(x_j)\ln p(r_j|x_j,\kappa) \mathrm{d}x_j\right\}\\ 
			&\propto \kappa \exp\left\{-\kappa \left(\left| r_j -\tilde{x}_j\right|^2 + \tilde{\tau}_j\right)\right\}.
		\end{aligned}
	\end{equation} 
	
	Combining the prior message $p(\kappa)$ and $\nu_{f_{r_j} \rightarrow \kappa} (\kappa)$, the belief of $\kappa$ is given by
	\begin{equation}
		\begin{aligned}
			b(\kappa) 
			\propto \kappa^{N_t-1}\exp\left\{-\kappa \sum_{j=1}^{N_t} \left(\left| r_j - \tilde{x}_j\right|^2 + \tilde{\tau}_j\right)\right\}.
		\end{aligned}
		\label{post_beta1}
	\end{equation}
	It is observed that the belief $b(\kappa)$ obeys the Gamma distribution with shape parameters $N_t$ and $\sum_{j=1}^{N_t}\left(\left| r_j - \tilde{x}_j\right|^2 + \tilde{\tau}_j\right)$. Thus, the approximate posterior mean $\tilde{\kappa}$ is given by
	\begin{equation}
		\tilde{\kappa} = \frac{N_t}{\sum_{j=1}^{N_t}\left(\left| r_j - \tilde{z}_j\right|^2 + \tilde{\tau}_j\right)}.
		\label{post_beta}
	\end{equation} 
	
	\begin{algorithm}
		\renewcommand{\algorithmicrequire}{\textbf{Input:}}
		\renewcommand{\algorithmicensure}{\textbf{Output:}}
		\caption{Proposed Bayesian Refinement Algorithm}
		\begin{algorithmic}[1]
			\Require $\mathbf{r}$.
			\Statex \textbf{Initialize:} $\hat{\tau}^{(0)}=1$, $\hat{\mathbf{x}}=\mathbf{0}$, $\hat{\boldsymbol{\gamma}}=\mathbf{1}, \hat{\kappa} = 1$, and $t=0$. 
			\While{the stopping criterion is not met}
			\State Update $\pi_j^{\mathrm{out}}$ according to (\ref{pi_out});
			\State Update $\psi^f_j$ and $\psi^b_j$ according to (\ref{psi_f}) and (\ref{psi_b});
			\State Update $\pi^{\mathrm{in}}_j$ according to (\ref{pi_in});
			\State Update $\tilde{\upsilon}_j$ according to (\ref{hat_ga});
			\State Update $\tilde{x}_j$ and $\tilde{\tau}_x^j$ according to (\ref{hat_x});
			\State Update $\tilde{\kappa}$ according to (\ref{post_beta});
			\State $t = t+1$.
			\EndWhile
			\Ensure $\tilde{\mathbf{x}}^{t+1}$.
		\end{algorithmic}
		\label{UAMP_SBL_MRF}
	\end{algorithm}
	
	The proposed channel refinement algorithm can be organized in a more succinct form, which is summarized in Algorithm \ref{UAMP_SBL_MRF} can be terminated when it reaches a maximum number of iteration or the difference between the estimates of two consecutive iterations is less than a threshold. 
	Once the global angular-delay channel $\tilde{\mathbf{x}}$ is obtained, the global spatial-frequency channel can be reconstructed as $\tilde{\mathbf{H}} = \mathbf{F}_{\mathrm{A}}^{\mathrm{F}}\tilde{\mathbf{H}}_{\mathrm{T}}\mathbf{F}_{\mathrm{D}}$, where $\tilde{\mathbf{H}}_{\mathrm{T}}$ is obtained by the reshape of $\tilde{\mathbf{x}}$.
	\begin{remark}
		The refinement scheme in \cite{De_MIMO_2} relies solely on angular-delay sparsity, ignoring coefficient correlations. It identifies significant coefficients and suppresses others using soft-threshold window. Additionally, the method depends heavily on precise covariance matrix knowledge, limiting its practical applicability.
		In contrast, this work introduces a Bayesian denoising approach that jointly exploits angular-delay sparsity and coefficient correlations via a Markov chain-based hierarchical prior, enhancing refinement performance.
	\end{remark}
	\section{Simulation Results}
	\label{section6}
	\begin{table}
		\renewcommand\arraystretch{1}
		\centering
		\caption{Simulation Parameters}
		\setlength{\tabcolsep}{1mm}{
			\begin{tabular}{l c}
				\toprule [1pt]
				\makecell[l]{Notations}& Parameters\\ 
				\midrule [0.5pt]
				\makecell[l]{Number of BS antenna $N_{\mathrm{t}}$} &128 \\
				\makecell[l]{Number of RF chains $N_{\mathrm{RF}}$} &1 \\
				\makecell[l]{Number of subarrays $M$} &4 \\
				\makecell[l]{Carrier frequency $f_c$}    &30GHz\\
				\makecell[l]{Number of pilot carriers $K$} &16 \\
				\makecell[l]{System bandwidth $f_s$} &1.6GHz \\
				\makecell[l]{Number of channel path $G$} &4\\
				\makecell[l]{Angle of arrival $\vartheta_l$} &$\mathcal{U}(-\pi/2, \pi/2)$\\
				\makecell[l]{Distance between BS and UE or scatters $r_l$} &[10, 50]m\\
				\makecell[l]{Number of iteration of SBL-GNNs $T$} & 5\\
				\makecell[l]{Round of MP of GNNs $L$} & 3\\
				\makecell[l]{Output message dimension of variable nodes $N_{\mathrm{u}}$} & 8\\
				Number of hidden layers $N_{\mathrm{h}_1}$ and $N_{\mathrm{h}_2}$ & 64, 32\\
				\bottomrule[1pt]
			\end{tabular}
		}
		\label{Parameters}
	\end{table}
	
	In this section, we evaluate the performance of the proposed SBL-GNNs and channel refinement algorithms. In the simulation, the key parameters are provided in Table \ref{Parameters}. In particular, we consider NMSE as performance metrics, which is defined as $\mathrm{NMSE} \triangleq {\lVert \hat{\mathbf{H}}_m (\tilde{\mathbf{H}}) - \mathbf{H}_m (\mathbf{H}) \rVert^2_{\mathrm{F}}}/{\lVert \mathbf{H}_m (\mathbf{H}) \lVert^2_{\mathrm{F}}}$,
	where $\mathbf{H}_m (\mathbf{H}) $ and $\hat{\mathbf{H}}_m (\tilde{\mathbf{H}})$ are the true channel and estimated channel, respectively. In addition, the received SNR is defined as $10\log_{10}\left(\lVert\mathbf{WH} \rVert^2_{\mathrm{F}}/ \lVert\mathbf{N} \rVert^2_{\mathrm{F}}\right)$.
	To evaluate channel estimation performance in LPUs, we compare our proposed SBL-GNNs algorithm with the following benchmarks, which adopt different sparse prior models:
	
	\begin{itemize}
		\item StdSBL \cite{StdSBL}: 
		The standard SBL algorithm, implemented
		within an EM framework.
		The algorithm employs a two-layer Gaussian-Gamma
		hierarchical prior model, where the posterior estimates are updated in the E-step, and the hyperparameters
		of the prior model are updated in the M-step, as shown in (\ref{update1})-(\ref{update5}).
		\item PC-SBL \cite{PC_SBL}: 
		A variant of the StdSBL framework that
		incorporates a pattern-coupled Gaussian prior model to
		exploit the block sparsity inherent in signals.
		\item VSP \cite{Variance_State_1}: 
		A variant of the StdSBL framework that employs a 4-connected MRF-based hierarchical prior model to
		capture the block sparsity of signals.
		\item UAMP-MRF \cite{tang2024spatial}: 
		Unity approximate message passing algorithm with a 4-connect MRF prior to characterize the dependencies among neighboring coefficients.
		\item Centralized. Compared with sparse Bayesian learning  based algorithms, the UAMP-MRF has the lowest computational complexity. Therefore, we consider that the centralized channel estimation is realized by UAMP-SBL algorithm with $M=1$.
	\end{itemize}
	
	\renewcommand{\arraystretch}{1.25}
	\begin{table}
		\centering
		\caption{{Computational Complexity Comparison}}
		\begin{tabular}{|c|c|}
			\hline
			{Algorithms}         & {Computational complexity $\mathcal{O}(\cdot)$} \\ \hline
			{StdSBL}                 &    {$\mathcal{O}(N^3K^3T_1)$}                           \\ \hline
			{PC-SBL}              &  {$\mathcal{O}(N^3K^3T_2)$}                           \\ \hline
			{VSP}              & {$\mathcal{O}(N^3K^3T_3)$}                          \\ \hline
			{UAMP}              & {$\mathcal{O}(K^3N^2PN_{\mathrm{RF}}T_4)$}                  
			\\ \hline
			{Proposed SBL-GNNs}        & {\tabincell{c}{$\mathcal{O}(K^3N^2PN_{\mathrm{RF}}T+K^2PN_{\mathrm{RF}}NT$ \\ $+N^3K^3T+LN_{\mathrm{u}}N_{\mathrm{h}_1}N_{\mathrm{h}_2}T)$}} \\       
			\hline
		\end{tabular}
		\label{Comparsion}
	\end{table}
	
	{The computational complexities are detailed in Table \ref{Comparsion}, where $T_1$, $T_2$, $T_3$, $T_4$, and $T$ denote the number of iterations required for different algorithms, respectively. From Table \ref{Comparsion}, it can be observed that, for each iteration, the StdSBL, PC-SBL, VSP, and the proposed SBL-GNNs algorithms have comparable computational complexities due to the involvement of matrix inversion operations. In contrast, the UAMP algorithm, which only involves matrix product operations, exhibits the lowest computational complexity per iteration.
	Notably, owing to the powerful capability of GNNs in capturing the correlation and sparsity of the channel, the SBL-GNNs algorithm requires significantly fewer iterations to achieve satisfactory performance, i.e., $T \ll T_1, T_2, T_3, T_4$. 
	Consequently, the proposed SBL-GNNs algorithm exhibits a lower overall computational complexity, as will be validated later.}
	
	\begin{figure}
		\centering
		\includegraphics[width=0.4\textwidth]{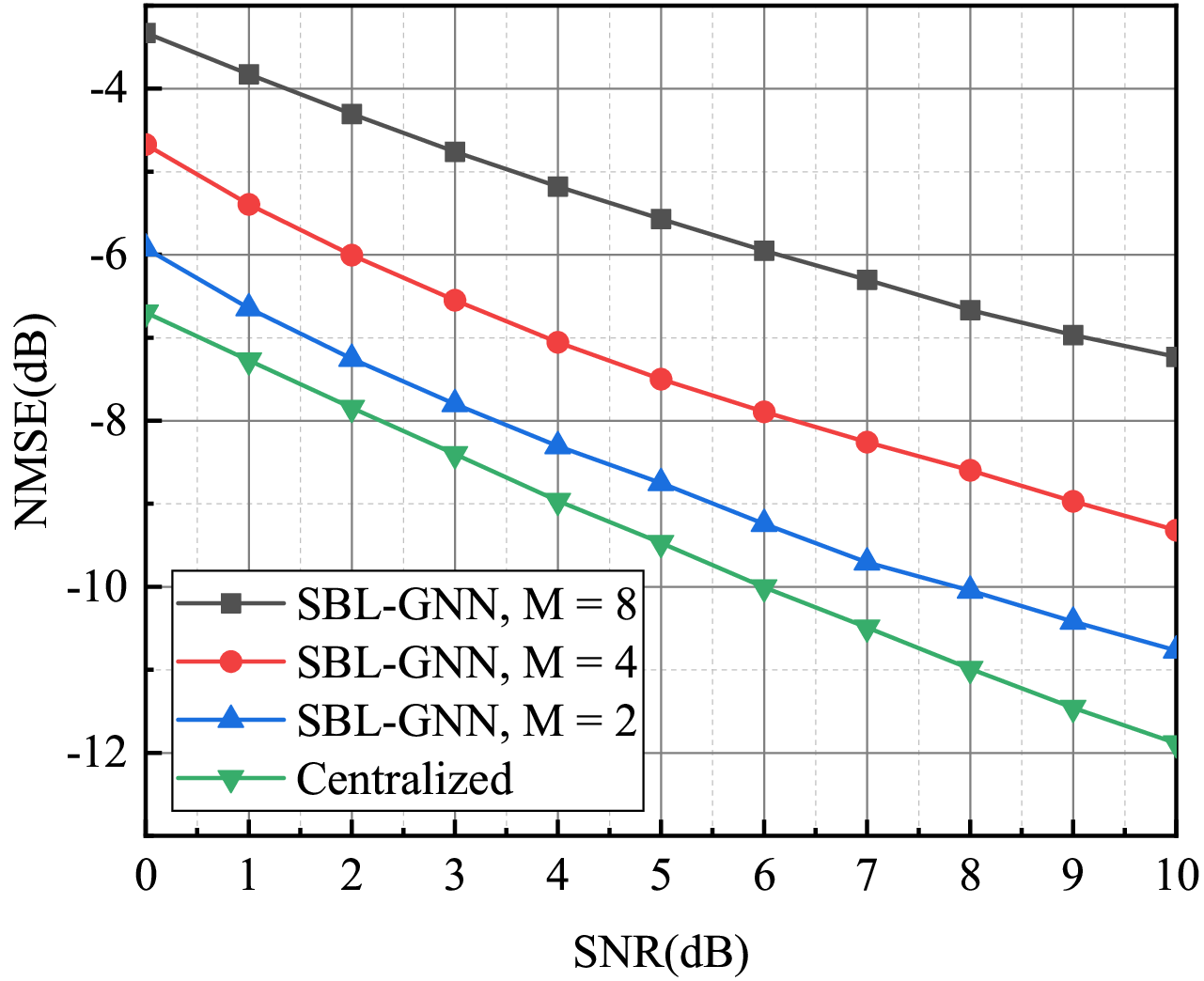}
		\caption{NMSE of SBL-GNNs  versus subarray number.}
		\label{performance_num}
	\end{figure}
	\vspace{-1em}
	\subsection{Parameter Selection and Computational Complexity}
	Fig. \ref{performance_num} illustrates the NMSE performance of the SBL-GNNs algorithm for various subarray configurations, with $P=N/2$. 
	It is evident that the centralized algorithm significantly outperforms the SBL-GNNs algorithm. 
	This disparity arises since the decentralized SBL-GNNs algorithm does not fully exploit the channel correlations among different subarrays. 
	Additionally, the NMSE gap between the centralized algorithm and the SBL-GNNs algorithm widens as the number of subarrays increases. However, it is important to note that the computational complexity of the centralized algorithm is substantially higher than that of the SBL-GNNs algorithm.
	To illustrate this observation more intuitively, Fig.~\ref{Run_time} presents the runtime of different algorithms. For fairness, all algorithms are executed on an i7-9700 CPU.
	Fig. \ref{Run_time1} demonstrates that the runtime of the SBL-GNNs algorithm is shorter than that of traditional algorithms. This improvement is attributed to the powerful capability of GNNs in capturing the structure of channel characteristics. Specifically, the SBL-GNNs algorithm requires only a few iterations (we set $T=5$ in the simulations) to achieve satisfactory performance, whereas traditional algorithms often need hundreds of iterations to ensure convergence. Consequently, this significant reduction in the number of iterations substantially decreases the runtime of the SBL-GNNs algorithm.
	\begin{figure}
		\centering
		\subfigure[Algorithm runtime in LPUs.]{
			\includegraphics[width=0.4\textwidth]{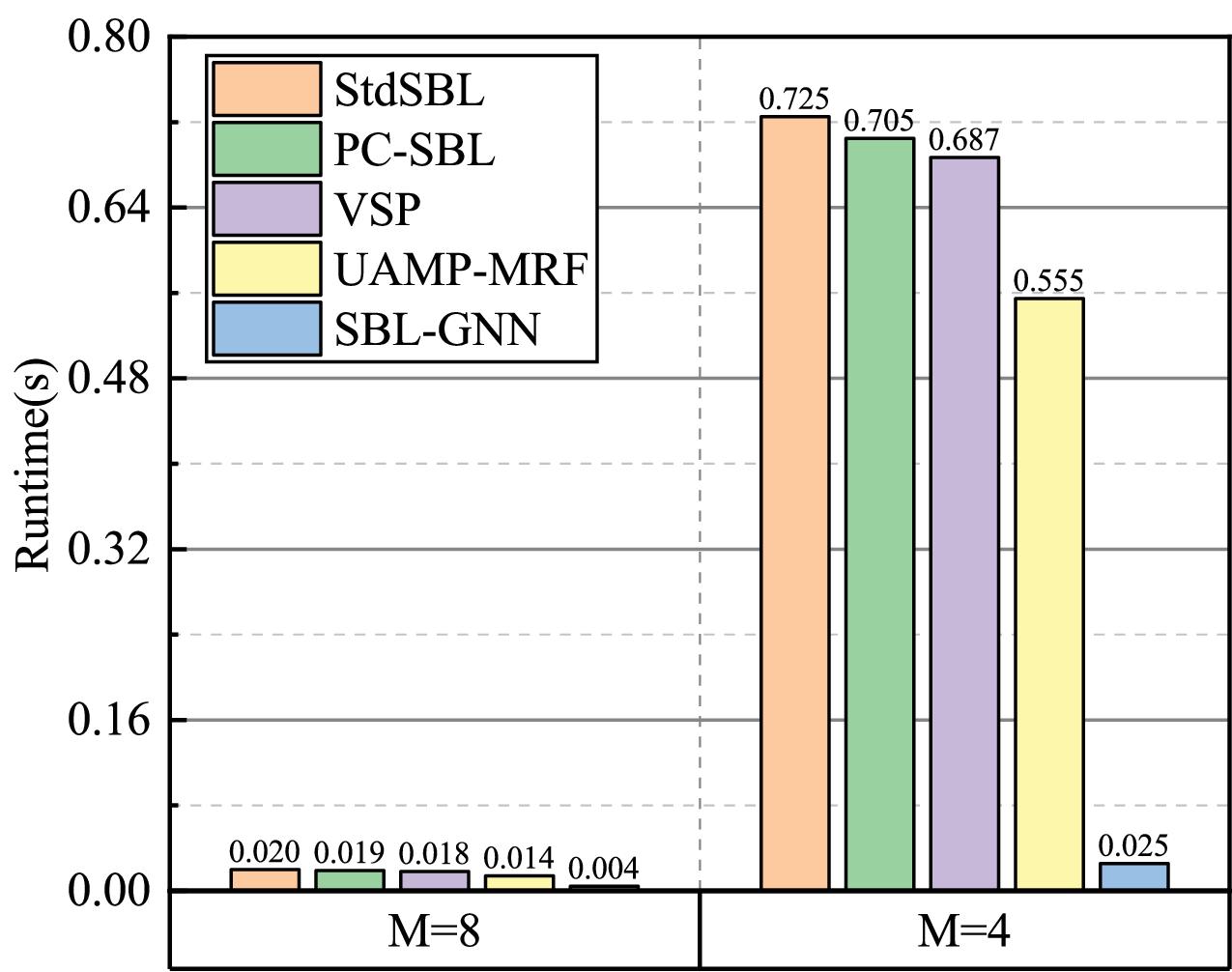}
			\label{Run_time1}
		} 
		\subfigure[Overall algorithm runtime]{
			\includegraphics[width=0.4\textwidth]{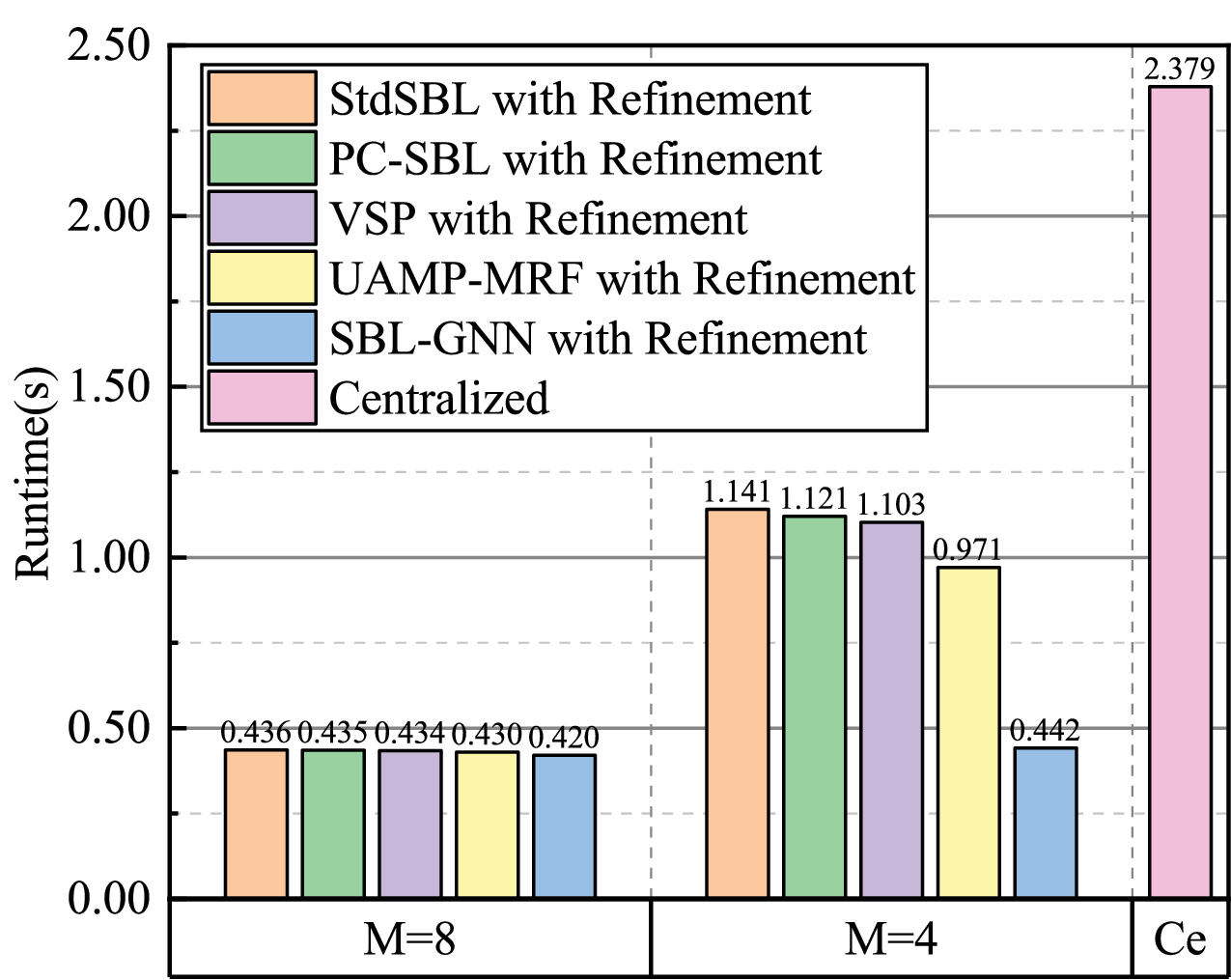}
			\label{Run_time2}
		}
		\caption{Computational complexity evaluation.}
		\label{Run_time}
	\end{figure}
	
	Similarly, Fig. \ref{Run_time2} provides the runtime of the overall algorithm. It can be seen that the runtime of the centralized algorithm is higher than that of all decentralized schemes due to the computational burden of high-dimensional matrix operations. In contrast, the proposed SBL-GNNs with channel refinement exhibits the shortest runtime in all considered scenarios, underscoring the computational efficiency of the proposed estimation scheme.
	Additionally, it is observed that the runtime of the overall SBL-GNNs algorithm with refinement is comparable when $M=8$ and $M=4$. {This is because the total runtime is determined by the combined execution time of Algorithm \ref{GNN_SBL_Algorithm} and Algorithm \ref{UAMP_SBL_MRF}, under the assumption of wired fronthaul transmission and perfect synchronization. Notably, since Algorithm \ref{UAMP_SBL_MRF} exhibits a significantly higher computational cost compared to Algorithm \ref{GNN_SBL_Algorithm}, the overall runtime is predominantly influenced by the Bayesian refinement process. As a result, variations in $M$ have a limited impact on the total runtime.
	Considering both estimation performance and computational complexity, we adopt $M=4$ as a practical choice for subsequent simulations.}
	\vspace{-1em}
	\subsection{Performance Evaluation of the SBL-GNNs Algorithm}
	\begin{figure}
		\centering
		\subfigure[NMSE versus SNR.]{
			\includegraphics[width=0.4\textwidth]{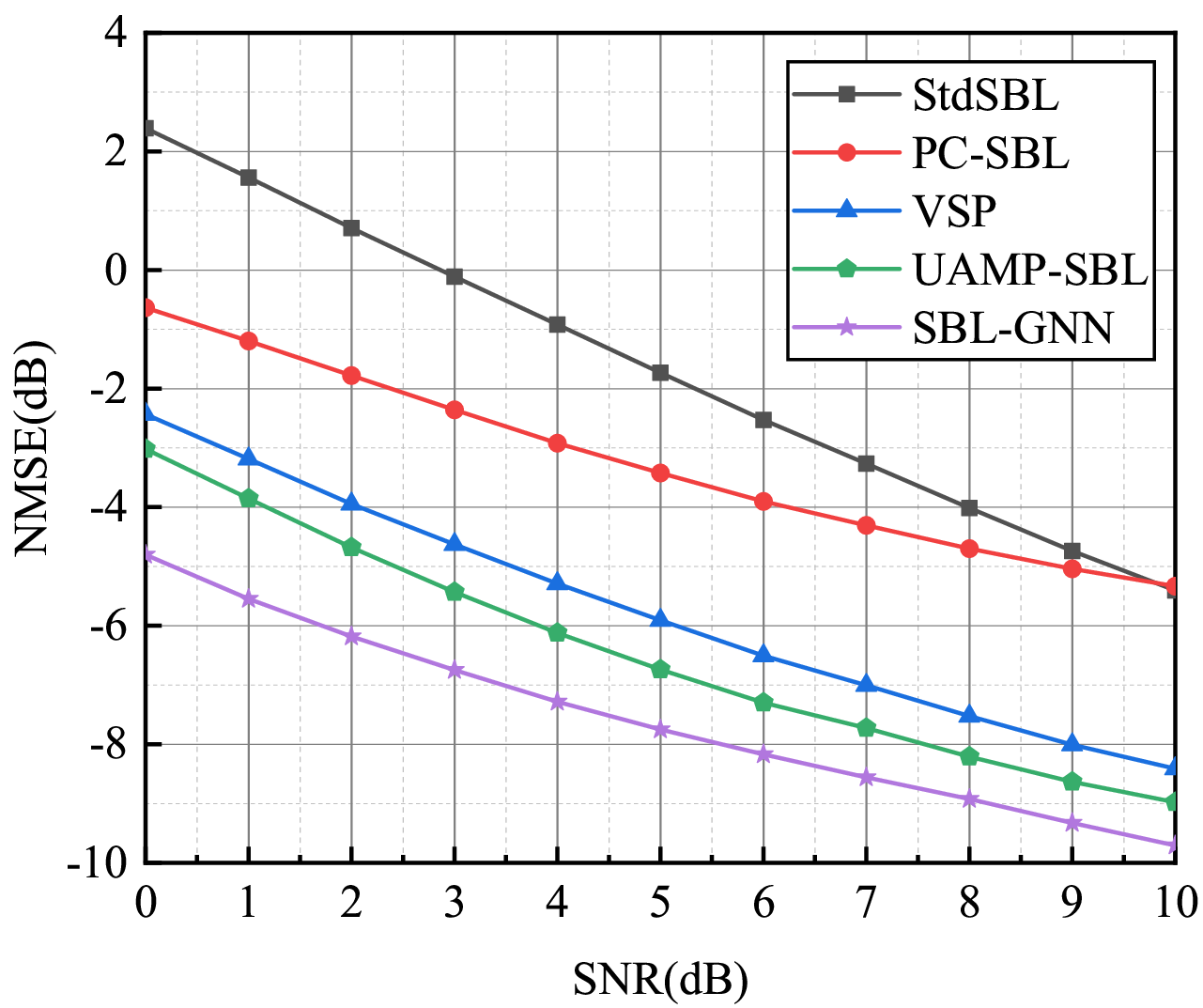}
			\label{SNR1}
		}
		\subfigure[NMSE versus pilot symbols.]{
			\includegraphics[width=0.4\textwidth]{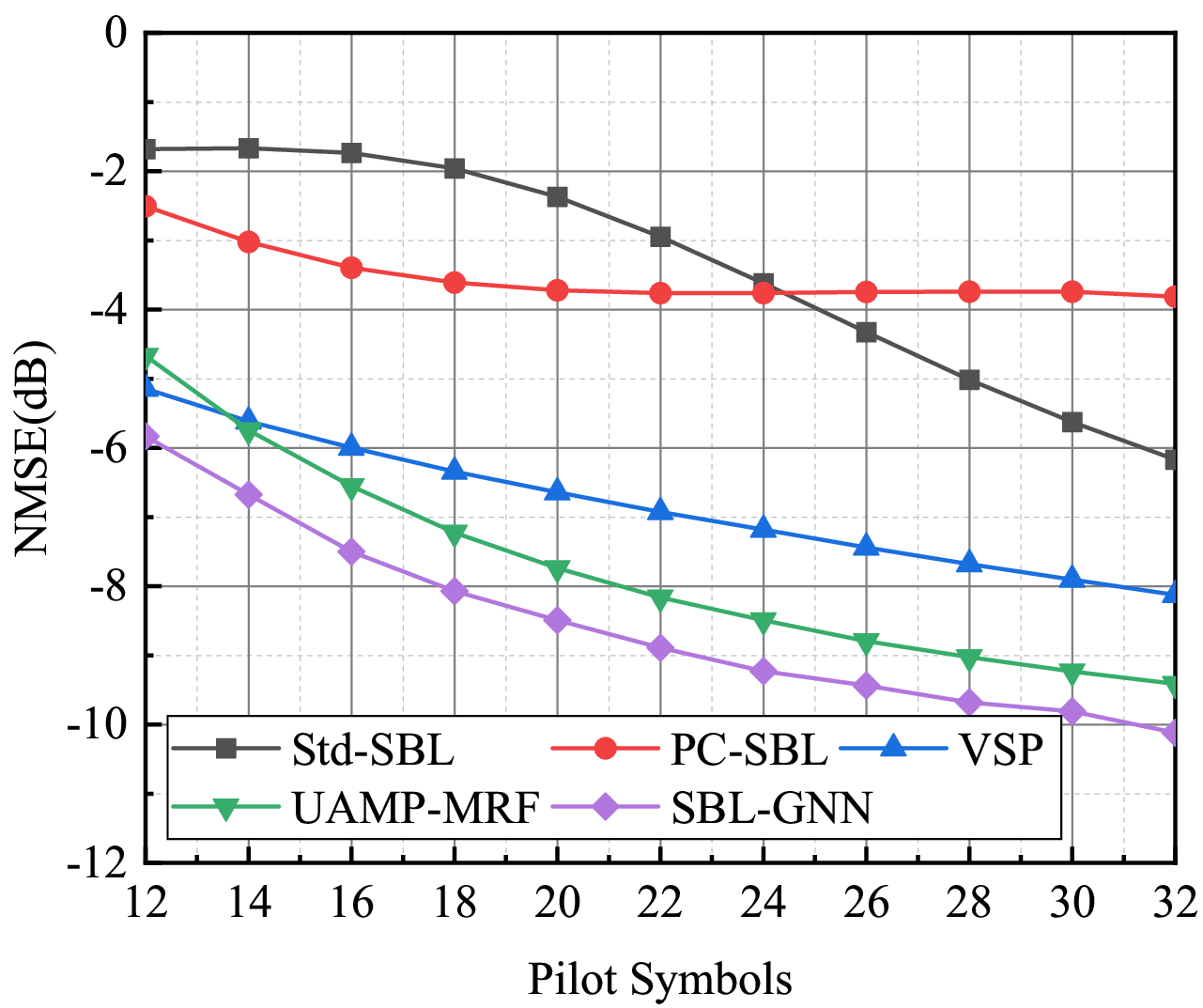}
			\label{Pilot1}
		}
		\caption{Performance evaluation of the SBL-GNNs algorithm.}
		\label{Decentralized}
	\end{figure}
	
	In this subsection, we investigate the NMSE performance of different sparse recovery algorithms executed on the LPUs. Fig. \ref{SNR1} examines the NMSE performance of various algorithms versus SNR with $M=4$, $N=32$, and $P=16$. It is evident that the StdSBL algorithm exhibits the worst estimation performance across all considered SNR ranges. 
	This disparity arises since the StdSBL algorithm employs a two-layer sparse prior, assuming independence among different angular and delay domain coefficients, which prevents the full exploitation of correlations or clustered sparsity among the coefficients. 
	In contrast, algorithms based on clustered-sparse priors significantly outperform the StdSBL algorithm.
	
	Specifically, the PC-SBL algorithm captures the dependencies among channel coefficients using a pattern-coupled hierarchical Gaussian prior. 
	The VSP and UAMP-MRF algorithms employ a 4-connected MRF-based three-layer sparse prior model to capture the structured sparsity in the angular and delay domains. 
	Notably, compared to the pattern-coupled prior, the MRF-based prior can more flexibly capture the clustered sparsity of channels due to the 4-connect structure. Consequently, the VSP and UAMP-MRF algorithms outperform the PC-SBL algorithm across the entire observation range.

	Additionally, it is observed that the proposed SBL-GNNs algorithm demonstrates superior performance compared to all benchmarks within the considered SNR ranges. This superiority arises, from the fully-connected MRF structure, which can directly capture the dependencies between any channel coefficients in the angular and delay domain. On the other hand, the model and data-driven scheme plays a crucial role in the performance improvement. Specifically, the model-based component of this scheme, i.e., the observation module, ensures the interpretability of the SBL-GNNs algorithm. Meanwhile, the GNNs module, with its powerful capability in capturing correlations among nodes and feature extraction, provides more precise prior information for the observation module. Therefore, the proposed SBL-GNNs algorithm outperforms the benchmarks.
	
	Furthermore, the performance comparison of the proposed SBL-GNNs algorithm and the benchmarks at different pilot symbols $P$ is investigated in
	Fig.  \ref{Pilot1} with $SNR=5$dB. The number of pilot symbols $P$ increases from 12 to 32, so the compressive ratio
	$PN_{\mathrm{RF}}/N$ correspondingly increases from 0.375 to 1. The estimation performance of all algorithms improves with the increase in the number of pilot symbols. Notably, the proposed SBL-GNNs algorithm consistently demonstrates superior estimation performance across all considered lengths of pilot symbols. Therefore, the SBL-GNNs algorithm provides a low pilot overhead estimation scheme while improving estimation performance.
	\vspace{-1em} 
	\subsection{Performance Evaluation of the Refinement Algorithm}
	\begin{figure}
		\centering
		\includegraphics[width=0.4\textwidth]{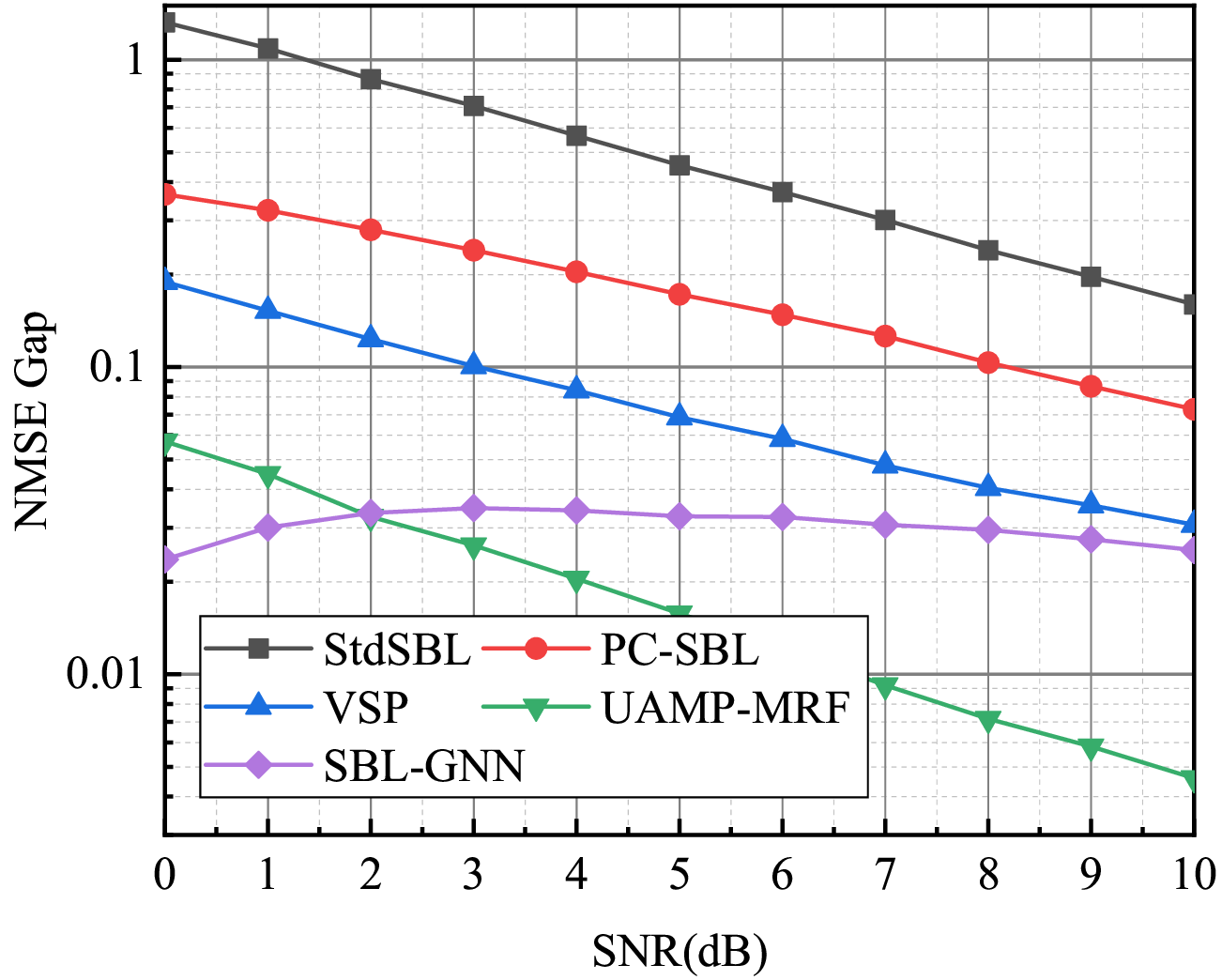}
		\caption{NMSE gap versus SNR.}
		\label{Gap}
	\end{figure}
	
	Then, we evaluate the NMSE performance of the channel refinement on the CPU. Fig. \ref{Gap} examines the NMSE gap, which is defined as $\mathrm{NMSE}^{\mathrm{De}}_{c} - \mathrm{NMSE}^{\mathrm{Re}}_{c}$, where $c\in\left\{\text{StdSBL}, \text{PC-SBL}, \text{VSP}, \text{UAMP-MRF}, \text{SBL-GNNs}\right\}$, $\mathrm{NMSE}^{\mathrm{De}}_{c}$ denotes the estimation performance of the decentralized schemes, and $\mathrm{NMSE}^{\mathrm{Re}}_{c}$ denotes the estimation performance of the decentralized estimation with central refinement. As anticipated, the estimation performance of all algorithms improves after refinement. Specifically, the StdSBL algorithm experiences the most significant performance improvement due to the Markov chain structure fully exploiting the dependencies among angular and delay domain coefficients, which are completely overlooked in the decentralized StdSBL algorithm. Similar improvement is observed in the PC-SBL scheme with channel refinement.
	For the VSP and UAMP-MRF schemes with refinement, the level of performance improvement is less than that of StdSBL and PC-SBL with refinement. This is because the decentralized VSP and UAMP-MRF algorithms already adopt an MRF-based prior to capture the dependencies among angular and delay domain channels within the subarray, resulting in a relatively limited performance gain from central refinement.
	Notably, unlike these benchmarks, the improvement of the decentralized SBL-GNNs algorithm through channel refinement is very consistent across different SNR levels. Therefore, the variation of the NMSE gap with SNR is relatively flat.
	
	\begin{figure}
		\centering
		\subfigure[NMSE versus SNR.]{
			\includegraphics[width=0.4\textwidth]{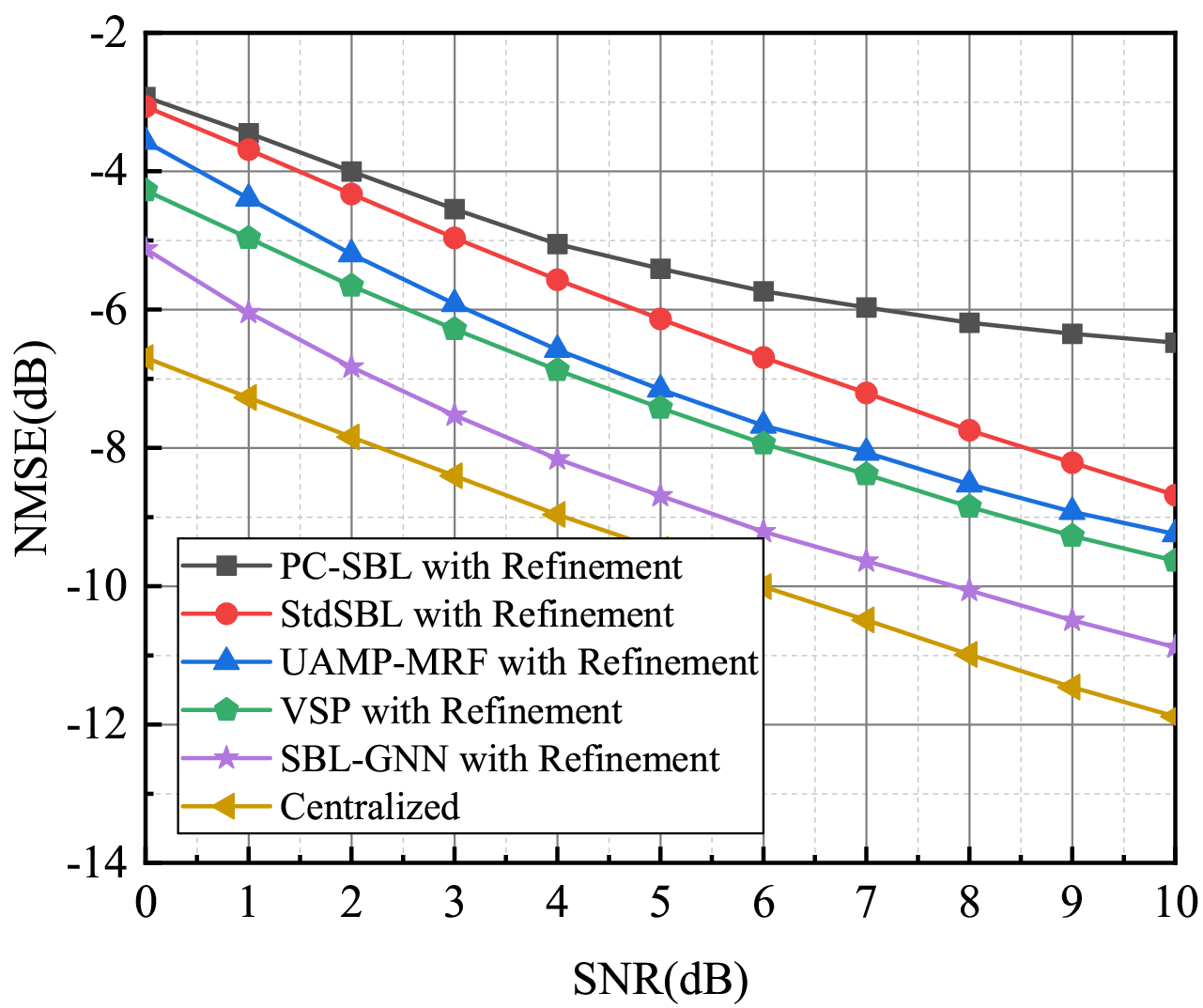}
			\label{SNR2}
		}
		\subfigure[NMSE versus pilot symbols.]{
			\includegraphics[width=0.4\textwidth]{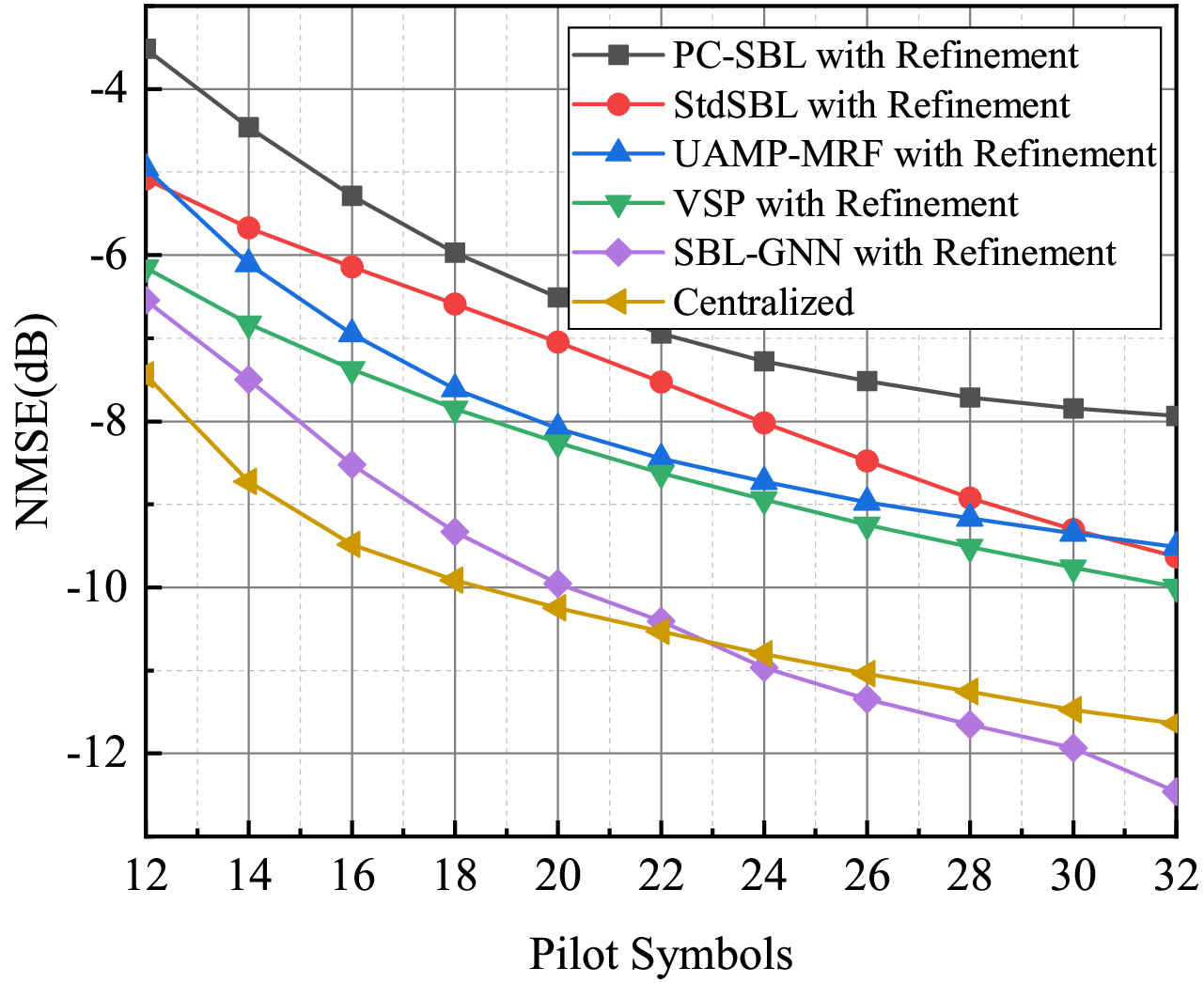}
			\label{Pilot2}
		}
		\label{Centralized_fig}
		\caption{Performance evaluation for overall algorithms.}
	\end{figure}
	
	To further evaluate the performance of the overall algorithm, Fig. \ref{SNR2} plots the NMSE performance of the overall algorithm versus SNR with $P=16$. It can be seen that, owing to the centralized processing, the centralized algorithms achieve the best estimation performance, serving as a lower bound. Meanwhile, the performance of the proposed decentralized SBL-GNNs with refinement is closest to that of the centralized algorithm, with a gap of about 1 dB. In contrast, the performance of the best algorithm in the benchmark is approximately 2 dB less than that of the centralized algorithm.
	A similar superiority is observed in Fig. \ref{Pilot2}, where, when the number of pilot symbols exceeds 25, the performance of the proposed algorithm even surpasses that of the centralized algorithm. Therefore, the proposed estimation algorithm provides a low-overhead estimation scheme while ensuring robust estimation performance.
	\setcounter{equation}{0}
	\renewcommand\theequation{A.\arabic{equation}}
	\begin{figure*}
		\normalsize
		\setcounter{MYtempeqncnt}{\value{equation}}
		\setcounter{equation}{1}
		\begin{align}
			\label{likehood} p(\boldsymbol{\mu}_m^{(t)}| \mathbf{x}_m)  &\propto \exp\left\{-\left(\mathbf{x}_m-\boldsymbol{\mu}_m^{(t)}\right)^{\mathrm{H}}\left(\boldsymbol{\Sigma}_m^{(t)}\right)^{-1}\left(\mathbf{x}_m-\boldsymbol{\mu}_m^{(t)}\right)\right\}\\
			\notag &=\prod_{n \in \mathcal{V}}\exp\left\{2\mathrm{Re}\left((\boldsymbol{\mu}_m^{(t)})^{\mathrm{H}}\mathbf{v}_{m,n}^{(t)}{x}_{m,n}\right)-v^{n,n}_m(t)\lvert x_{m,n} \rvert^2\right\}
			\prod_{n \in \mathcal{V}}\prod_{k\neq n}\exp\left\{-v^{n,k}_m(t) x_{m,n}x_{m,k}^*\right\}\\
			\notag &=\prod_{n \in \mathcal{V}}\psi_{\mathrm{like}}\left(x_{m,n}\right)\prod_{n \in \mathcal{V}}\prod_{k\neq n}\phi_{\mathrm{like}}\left(x_{m,n}, x_{m,k}\right),\\
			\label{prior2} p\left(\mathbf{x}_m; \boldsymbol{\gamma}_m, \mathbf{s}_m\right) &= \prod_{n \in \mathcal{V}}p\left(x_{m,n}|\gamma_{m,n}\right)p\left(\gamma_{m,n}|s_{m,n}\right) \psi\left(s_{m,n}\right)\prod_{n \in \mathcal{V}}\prod_{k\neq n}\phi\left(s_{m,n}, s_{m,k}\right)\\
			\notag &=\prod_{n \in \mathcal{V}}\psi_{\mathrm{pri}}\left(\alpha_m, x_{m,n}\right)\prod_{n \in \mathcal{V}}\prod_{k\neq n}\phi_{\mathrm{pri}}\left(\beta_m, x_{m,n}, x_{m,k}\right),
		\end{align}
		\hrulefill
		\vspace*{4pt}
	\end{figure*}
	\vspace{-1em}
	\section{Conclusion}
	\label{section7}
	In this paper, we have addressed the channel estimation problem for XL-MIMO systems with a hybrid analog-digital architecture, implemented within a DBP framework and a star topology. To overcome the limitations of existing centralized and fully decentralized estimation schemes, we have proposed a novel two-stage channel estimation scheme integrating local sparse reconstruction and global fusion with refinement. For local reconstruction, we have introduced an efficient SBL approach enhanced by GNNs, where E-step of the traditional SBL algorithm is retained, and M-step is updated with GNNs to effectively capture dependencies among channel coefficients. For the global refinement, we have proposed a Markov chain-based hierarchical prior model and a corresponding variational message passing algorithm to exploit the sparsity and correlations in the channel, improving estimation accuracy and efficiency. 
	Simulation results validate the effectiveness of the proposed algorithms, demonstrating that the SBL-GNNs algorithm achieves low complexity with satisfactory performance. Additionally, the overall computational complexity of the two-stage scheme is significantly lower than that of the centralized approach, while maintaining comparable estimation performance.
	\vspace{-1em}
	\section*{Appendix A}
	According to the equivalent AWGN model in (\ref{AWGN2}), $p(\mathbf{x}_m,\boldsymbol{\gamma}_m,\mathbf{s}_m|\boldsymbol{\mu}_m^{(t)})$ can be given by
	\begin{equation}
		p(\mathbf{x}_m,\boldsymbol{\gamma}_m,\mathbf{s}_m|\boldsymbol{\mu}_m^{(t)})
		\propto p(\boldsymbol{\mu}_m^{(t)}|\mathbf{x}_m)p\left(\mathbf{x}_m; \boldsymbol{\gamma}_m,\mathbf{s}_m\right),
		\label{posterior}
	\end{equation}
    where $p(\boldsymbol{\mu}_m^{(t)}|\mathbf{x}_m)$ and $p\left(\mathbf{x}_m; \boldsymbol{\gamma}_m,\mathbf{s}_m\right)$ denote the likelihood and prior distributions, respectively. In the following, we first factorize the likelihood and prior distributions. Then, we further consider the factorization of the posterior distribution. 
    In the terms of $p(\boldsymbol{\mu}_m^{(t)}|\mathbf{x}_m)$, we have (\ref{likehood}),
    where $\psi_{\mathrm{like}}(x_{m,n}) \triangleq \exp\{2\mathrm{Re}((\boldsymbol{\lambda}_m^{(t)})^{\mathrm{H}}\mathbf{v}_{m,n}^{(t)}{x}_{m,n})-v^{n,n}_m(t)\lvert x_{m,n} \rvert^2\}$ denotes the self-potential of $x_{m,n}$; $\phi_{\mathrm{like}}(x_{m,n}, x_{m,k}) \triangleq \exp\{-v^{n,k}_m(t) x_{m,n}x_{m,k}^{*}\}$ denotes the pair-potential between $x_{m,n}$ and $x_{m,k}$; $\mathbf{v}_{m,n}^{(t)}$ and $v^{n,k}_m(t)$ indicate the $n$-th column and $(n,k)$-th entry of $\mathbf{V}_m^{(t)} = (\boldsymbol{\Sigma}_m^{(t)})^{-1}$, respectively.
    
    {As for $p(\mathbf{x}_m; \boldsymbol{\gamma}_m,\mathbf{s}_m)$, according to (\ref{prior})-(\ref{3rd_layer}), we have (\ref{prior2}),
    where $\psi_{\mathrm{pri}}(\alpha_m, x_{m,n})$ and $\phi_{\mathrm{pri}}(\beta_m, x_{m,n}, x_{m,k})$ implicitly  denote pair-potential and self-potential functions for the prior probability distribution $p(\mathbf{x}_m; \boldsymbol{\gamma}_m,\mathbf{s}_m)$ with $\alpha_m$ and $\beta_m$ indicating the parameters of the MRF.
    Combining (\ref{likehood}) and (\ref{prior2}), (\ref{MRF}) can be obtained.}
	\bibliographystyle{IEEEtran}
	\bibliography{ref}
\end{document}